\documentclass[12pt,letterpaper]{article}

\usepackage[left=1in,right=1in,top=1.5in,bottom=1.5in]{geometry}
\usepackage{setspace}

\usepackage{amsmath,amssymb,amsthm,mathtools,bm}

\usepackage{graphicx}
\usepackage{adjustbox}
\usepackage{booktabs}
\usepackage{threeparttable}
\usepackage{longtable}
\usepackage{array}
\usepackage{multirow}
\usepackage{caption}
\usepackage{subcaption}
\usepackage{float}
\usepackage{placeins}
\usepackage{makecell}

\usepackage{tabularx}

\newcolumntype{L}[1]{>{\raggedright\arraybackslash}p{#1}}
\newcolumntype{Y}{>{\raggedright\arraybackslash}X}

\usepackage{enumitem}
\usepackage{xcolor}
\usepackage{appendix}
\usepackage{pdflscape}
\usepackage{csquotes}
\usepackage{bbm}

\usepackage{chngcntr}
\counterwithin{figure}{section}
\counterwithin{table}{section}

\newcommand{\nostar}{\phantom{$^{***}$}}

\usepackage[authoryear]{natbib}

\usepackage[colorlinks=true,linkcolor=blue,citecolor=blue,urlcolor=blue]{hyperref}
\usepackage[nameinlink,capitalize]{cleveref}

\newtheorem{proposition}{Proposition}
\newtheorem{assumption}{Assumption}

\theoremstyle{definition}
\newtheorem{definition}{Definition}
\newtheorem{remark}{Remark}

\AtBeginDocument{
	\setlength{\abovedisplayskip}{7pt plus 2pt minus 2pt}
	\setlength{\belowdisplayskip}{7pt plus 2pt minus 2pt}
	\setlength{\abovedisplayshortskip}{4pt plus 2pt minus 1pt}
	\setlength{\belowdisplayshortskip}{4pt plus 2pt minus 1pt}
}

\title{Overshooting the Coordinate:\\Where Factor Corrections Land on Characteristic Axes}

\author{Useong Shin\thanks{
		Sogang Business School, Sogang University (Seoul, Korea).\\
		ORCID: \href{https://orcid.org/0009-0003-0197-9003}{0009-0003-0197-9003}\\
		Email: \texttt{useong@sogang.ac.kr}
}}
\date{\today}

\begin{document}
	
	\pagenumbering{arabic}
	\setcounter{page}{1}
	
	\maketitle
	
	\begin{flushleft}
	\end{flushleft}
	
	\begin{abstract}
		Maximum-Sharpe spanning measures how factors expand the investment opportunity set, not where they leave pricing errors along the characteristics they are built to price. I develop a characteristic-axis diagnostic that traces model alphas across the full order induced by a predetermined characteristic. In 1967--2024 CRSP data, counterpart factors move pricing-error curves downward but land differently: profitability and momentum flatten, value overshoots on average but shifts across correction regimes, and investment crosses zero with weaker evidence. Sharpe gains do not determine these destinations, and q5's value non-rejections partly reflect wider null bands. Factor-model evaluation therefore requires both spanning and coordinate-specific pricing.
	\end{abstract}
	
	\newpage
	\section{Introduction}
	\label{sec:intro}
	
	Maximum-Sharpe spanning asks how much a factor expands the mean--variance opportunity set, not where it leaves pricing errors along the economic coordinate from which it was constructed. A low-dimensional model can therefore improve the attainable Sharpe ratio while leaving systematic alphas on a restricted return subspace. Spanning compares factor spaces without locating those errors \citep{BS17,BS18,FF18}, whereas GRS and characteristic-sorted tests evaluate supplied assets and inherit sensitivity to test-asset and bin choice \citep{GRS89,LNS10,GXZ25}. Neither uses the full within-characteristic order to show whether a factor correction stops short of, reaches, or overshoots zero.
	
	This paper develops a characteristic-axis integral diagnostic for that question. I fix a characteristic before observing test-period returns, sort stocks by it, and at each cutoff $p$ form a cap-weighted prefix portfolio. Subtracting equal exposure to the aggregate of the same valid subuniverse produces a zero-investment bridge $D^x_t(p)$ whose model alpha $\alpha^x_m(p)$ traces a pricing-error curve over the full order. Together with a zero-alpha aggregate gate, a zero bridge curve prices the prefix, tail, interval, and step-function returns generated by the axis; if the gate fails, the restriction applies conditionally to the internal bridge subspace. Signed area records direction, while $IAE$, $ISE$, and the sup norm measure absolute, concentrated, and local distortion. Conditional on the characteristic, order, and traversal measure, the path is generated mechanically rather than selected through a few bins.\footnote{The construction generalizes the cap-axis diagnostic of \citet{Shin26}; \Cref{sec:general} states the generalization.}
	
	I apply the diagnostic to value, operating profitability, investment, and momentum in 1967--2024 CRSP data. Counterpart-factor content generally shifts pricing-error curves downward, but the destinations differ. HML-based models push value through zero into significantly negative full-sample errors, with rank alphas near $-28$~bp; rolling 10-year estimates show that this overshoot varies through time rather than arising from one isolated period. Investment-factor models also cross zero, although the evidence is weaker. RMW and UMD instead bring profitability and momentum close to zero. Adding UMD to FF3, for example, moves the momentum rank alpha from $+92.6$ to $-7.1$~bp, with no functional rejecting at 5\%, providing a positive control in which factor correction aligns closely with the coordinate.
	
	These results are not produced by a uniformly more aggressive test. Across the four characteristic axes and the market-capitalization benchmark, conventional decile GRS and axis $IAE$ disagree in both directions. The axis diagnostic rejects FF5 and FF6 on value despite their passing decile GRS and the extreme-decile spread, because their negative errors extend through the interior of the order. Conversely, Carhart and FF6 fail conventional momentum tests but not rank area or $IAE$, because their remaining errors are localized and nonmonotone rather than cumulative. GRS retains broader scope over arbitrary finite asset sets and cross-anomaly restrictions; the axis diagnostic gives up that omnibus role for deeper, order-aware resolution within one prespecified coordinate.
	
	The evidence also separates frontier expansion from axis destination. FF5 explains about 76\% of the value-axis payoff and is nevertheless rejected on it, while candidate factors that produce large maximum-Sharpe gains need not bring the relevant coordinate close to zero. By the Frisch--Waugh--Lovell identity, the appraisal ratio determines the potential scale of a factor's contribution, but its signed residual alignment with the axis and the model's starting curve determine whether the enlarged model undercorrects, flattens, or overshoots. The destination also belongs to the traded factor implementation rather than to the characteristic in the abstract: removing, replacing, or repackaging counterpart factors can materially change the resulting curve even when the underlying signal is held fixed.
	
	Inference is model-specific because each nonlinear functional is calibrated to the residual covariance of that model's bridge regressions. A noisier residual curve produces wider null bands, so non-rejection can reflect weak resolution rather than favorable pricing; this distinction is especially important for q5 on value. Factor replication, conventional-decile comparisons, moving-block-bootstrap inference, finite-sample calibration, and subperiod analysis support the main conclusions, but the claims remain local. The aggregate gate and zero bridge curve price only the subspace generated by one chosen characteristic order, and empirical non-rejection means only that no distortion is detected at the implemented grid and resolution. Within those limits, the diagnostic reveals where factor correction lands along the coordinate it was built to price, complementing rather than replacing spanning and conventional specification tests.
	
	\section{Characteristic-Axis Pricing in Context}
	\label{sec:lit}
	
	\paragraph{Factor spanning and the missing diagnostic question.}
	Multifactor models explain overlapping return patterns through different factor sets and economic motivations \citep{FF93,Carhart97,FF15,HXZ15,HMXZ19,HXZ20,HMXZ21,HMXZ24}. For traded factors, relative mean--variance comparison is governed largely by factor spanning and the associated maximum Sharpe ratio rather than repeated tests on arbitrary common assets \citep{BS17,BS18,FF18}, while much SDF-relevant information may lie in a few return directions \citep{KNS18}. Taking spanning as given, I ask where factor correction leaves pricing errors along a prespecified economic coordinate. The appraisal ratio scales the potential correction, but signed residual alignment and the starting curve determine whether it stops short of zero, reaches it, or overshoots. Spanning evaluates frontier expansion; characteristic-axis pricing localizes errors on a restricted return subspace.
	
	\paragraph{Test assets and characteristic-axis diagnostics.}
	Specification tests such as GRS, Hansen--Jagannathan distance, and linear-SDF tests evaluate pricing errors on supplied assets \citep{GRS89,HJ97,Cochrane05}, so their conclusions depend on the chosen cross-section. Common factor structure, weak factors, correlated exposures, specification search, and construction choices can materially affect inference \citep{LNS10,GXZ25,LM90,HXZ20,JKP23}. Recent work therefore constructs more informative assets through managed portfolios or supervised partitions \citep{Kirby20,BPZ25}. The axis diagnostic instead fixes one characteristic, ordering direction, and traversal measure ex ante, mechanically generating the full prefix path rather than selecting bins or optimized splits. Each bridge removes equal exposure to the characteristic-valid aggregate, so the null is a zero pricing-error curve on the internal bridge subspace. Together with a zero-alpha aggregate gate, it prices the prefix, tail, interval, and step-function returns generated by the order. The procedure does not rank models globally; it provides order-aware resolution within one prespecified coordinate.
	
	\paragraph{Characteristics, factors, and traded implementations.}
	Many factors originate from firm characteristics: size and book-to-market underlie FF3, momentum captures return continuation, and investment and profitability enter later models \citep{FF92,FF93,JT93,Carhart97,TWX04,NM13,FF15,HXZ15}. I take no position on whether characteristics proxy for priced covariances or directly describe expected returns \citep{DT97}. Instead, I treat value, operating profitability, investment, and momentum as established economic coordinates and study how model alphas vary along their full orders \citep{FF08,HMXZ21}. The contribution is not to rediscover their average-return spreads, but to locate where factor correction leaves the returns those orders generate.
	
	The relevant correction belongs to a traded implementation, not to the characteristic in the abstract. Size-balanced and multidimensional sorts package signals while controlling size and diversifying firm-level noise \citep{FF93,FF15,HXZ15}, but are not designed to minimize a pricing-error curve on the underlying cap-weighted coordinate. Implementations based on similar signals can therefore differ in appraisal ratios, residual alignment, and destinations. More generally, sorted portfolios may combine priced and unpriced variation, alternative constructions need not span the same SDF directions, and anomaly results depend on breakpoints, screens, and definitions \citep{DMRS20,KN23,FF08,HXZ20,JKP23,SVV24}. An axis rejection thus concerns a particular traded construction rather than the characteristic or model family; \Cref{subsec:construction_sensitivity} tests this distinction by repackaging the same signals.
	
	\paragraph{Model-dependent resolution and the functional view.}
	Mechanical path construction does not make inference equally informative across models. Each nonlinear functional is calibrated to the bridge curve's model-specific residual covariance, so noisier models face wider null distributions and looser critical values. I therefore report model-specific 5\% thresholds and the rank-area $R^2_Y$, using the latter only as a resolution diagnostic; \Cref{subsec:robust_size_power} evaluates the resulting finite-sample asymmetry. The approach also relates to the functional interpretation of portfolio sorting. Whereas characteristic-sorted estimators and managed portfolios study expected-return functions \citep{CCFS20,Kirby20}, I estimate the bridge-alpha function $p\mapsto\alpha^x_m(p)$. Signed area records direction, $IAE$ absolute distortion, $ISE$ concentration, and $SUP$ the largest local error. The restriction is therefore zero alpha on the full characteristic-axis subspace, not on selected bins.
	
	\section{Generalizing the Axis Integral Diagnostic}
	\label{sec:general}
	
	The cap-axis diagnostic of \citet{Shin26} measures pricing errors inside the market portfolio as a function of capitalization rank. Its essential ingredients are more general: a characteristic fixed before returns are observed, a universe sortable by that characteristic, and an implementable aggregate portfolio on the same universe. The resulting characteristic-axis diagnostic turns the destination of a factor correction into a zero-curve restriction: it asks whether the pricing-error path remains on one side of zero, is flattened, or is pushed beyond it. This section defines that diagnostic and relates its curve shifts to maximum-Sharpe spanning.
	
	\subsection{Characteristic axis and bridge}
	\label{subsec:general-axis}
	
	Let $x$ be a stock-level characteristic observed at formation date $\tau(t)$, and let $\mathcal{I}^{x}_{\tau(t)}$ be the set of stocks for which it is valid. This set may be smaller than the investible market; accounting axes, for example, are defined on the valid CRSP--Compustat subuniverse. The membership rule and characteristic order must use only information available at formation.
	
	Sort stocks in descending order of $x_{i,\tau(t)}$ and assign each stock a strictly positive initial portfolio weight. A useful family is
	\begin{equation}
		\pi^x_{i,\tau(t)}
		\propto
		ME_{i,\tau(t)}^\rho,
		\qquad
		\rho\geq0,
		\label{eq:power_cap_weight}
	\end{equation}
	which includes cap weighting $(\rho=1)$ and equal weighting $(\rho=0)$. Initial weights are normalized to one and carried forward as buy-and-hold positions. The characteristic order remains fixed during the holding period, while the start-of-day portfolio values induce the current traversal measure. Suppressing this time dependence in the rank notation, let $u\in[0,1]$ denote cumulative current portfolio mass along the fixed order. Thus $p$ is a portfolio-mass cutoff, not necessarily a cross-sectional characteristic quantile.
	
	Let $r^x_t(u)$ denote the simple return at rank $u$. The corresponding characteristic-valid aggregate return and its excess return are
	\begin{equation}
		R^{A,x}_t
		=
		\int_0^1 r^x_t(u)\,du,
		\qquad
		\widetilde R^{A,x}_t
		=
		R^{A,x}_t-R^f_t.
		\label{eq:axis_aggregate_return}
	\end{equation}
	Under cap weighting, $R^{A,x}_t$ is the value-weighted aggregate of the valid-characteristic universe. It coincides with the investible market only when that valid universe coincides with the market; otherwise it remains an implementable subuniverse portfolio.
	
	For cutoff $p\in[0,1]$, define the prefix contribution $C^x_t(p)=\int_0^p r^x_t(u)\,du$ and subtract the same exposure to the aggregate:
	\begin{equation}
		D^x_t(p)
		=
		C^x_t(p)-pR^{A,x}_t.
		\label{eq:axis_bridge}
	\end{equation}
	The bridge is zero investment and closes at both endpoints,
	\begin{equation}
		D^x_t(0)=D^x_t(1)=0,
		\label{eq:axis_bridge_closed}
	\end{equation}
	and the corresponding tail bridge satisfies
	\begin{equation}
		\int_p^1r^x_t(u)\,du
		-
		(1-p)R^{A,x}_t
		=
		-D^x_t(p),
		\label{eq:axis_tail_bridge}
	\end{equation}
	so one prefix curve records the body--tail offset at every cutoff.
	
	For factor model $m$ with factor vector $f_{m,t}$, estimate
	\begin{equation}
		D^x_t(p)
		=
		\alpha^x_m(p)
		+
		\beta^x_m(p)'f_{m,t}
		+
		\varepsilon^x_{m,t}(p).
		\label{eq:axis_bridge_regression}
	\end{equation}
	The function $p\mapsto\alpha^x_m(p)$ is the characteristic-axis bridge-alpha curve. Its population zero-curve restriction is
	\begin{equation}
		H^x_0:
		\alpha^x_m(p)=0
		\quad
		\text{for all }p\in[0,1].
		\label{eq:axis_zero_curve_null}
	\end{equation}
	
	\subsection{Functional summaries and the rank-area portfolio}
	\label{subsec:axis_functionals}
	
	I summarize the bridge-alpha curve by
	\begin{align}
		SA^x_m
		&=
		\int_0^1\alpha^x_m(p)\,dp,
		\label{eq:axis_sa}
		\\
		IAE^x_m
		&=
		\int_0^1\left|\alpha^x_m(p)\right|\,dp,
		\label{eq:axis_iae}
		\\
		ISE^x_m
		&=
		\int_0^1\alpha^x_m(p)^2\,dp,
		\label{eq:axis_ise}
		\\
		SUP^x_m
		&=
		\sup_{p\in[0,1]}
		\left|\alpha^x_m(p)\right|.
		\label{eq:axis_sup}
	\end{align}
	$SA$ measures direction, $IAE$ total absolute distortion, $ISE$ quadratic concentration, and $SUP$ the largest local error. Reversing the characteristic order maps the curve into $-\alpha^x_m(1-p)$; it therefore changes the sign of $SA$ but leaves the other three functionals unchanged. In the empirical tables $ISE$ is reported as $\sqrt{ISE^x_m}$, so that all three nonlinear functionals are measured in the units of the curve itself; the square root is strictly monotone, so it leaves the tests in \Cref{subsec:finite_axis_implementation} unchanged and maps their critical values as $q_{95}(\sqrt{ISE})=\sqrt{q_{95}(ISE)}$.
	
	The signed area is the alpha of a single zero-investment return. Integrating \eqref{eq:axis_bridge} gives
	\begin{align}
		\int_0^1D^x_t(p)\,dp
		&=
		\int_0^1
		\left\{
		\int_0^pr^x_t(u)\,du
		-
		pR^{A,x}_t
		\right\}dp
		\nonumber\\
		&=
		\int_0^1
		\left(\frac12-u\right)
		r^x_t(u)\,du
		\;\equiv\;
		Y^x_t,
		\label{eq:axis_rank_area}
	\end{align}
	so that, by linearity of the alpha functional,
	\begin{equation}
		\alpha_m(Y^x)
		=
		\int_0^1\alpha^x_m(p)\,dp
		=
		SA^x_m.
		\label{eq:axis_rank_area_alpha}
	\end{equation}
	Thus $SA$ is the alpha of a long--short portfolio with linear rank weights $\frac12-u$: it is long the high-characteristic side and short the low-characteristic side, with intensity increasing toward the endpoints. Changing the traversal measure changes these portfolio weights but not the logic of the construction.
	
	\subsection{Restricted characteristic-axis pricing}
	\label{subsec:restricted_axis_pricing}
	
	For a fixed characteristic universe, order, and admissible measure, define
	\begin{equation}
		\mathcal{V}_x
		=
		\overline{
			\operatorname{span}
			\left\{
			\widetilde R^{A,x},
			D^x(p):p\in[0,1]
			\right\}
		}.
		\label{eq:axis_subspace}
	\end{equation}
	
	\begin{assumption}[Regularity]
		\label{ass:axis_regular}
		All test returns and factor returns have finite second moments, and $\operatorname{Var}(f_m)$ is nonsingular.
	\end{assumption}
	
	\begin{definition}[Admissible measure]
		\label{def:admissible_measure}
		A formation-date measure is \emph{admissible} if its initial weights use only information available at $\tau(t)$, place strictly positive mass on every valid stock, have finite total mass, and evolve thereafter only through the returns on the predetermined holdings. The power-of-cap family in \eqref{eq:power_cap_weight} is admissible for every $\rho\geq0$.
	\end{definition}
	
	\begin{proposition}[Restricted characteristic-axis pricing]
		\label{prop:restricted_axis_pricing}
		Under \Cref{ass:axis_regular}, model $m$ prices every return in $\mathcal{V}_x$ if and only if
		\begin{equation}
			\alpha_m(\widetilde R^{A,x})=0
			\quad\text{and}\quad
			\alpha_m(D^x(p))=0
			\quad
			\text{for all }p\in[0,1].
			\label{eq:restricted_axis_condition}
		\end{equation}
	\end{proposition}
	
	\begin{proof}
		For any square-integrable excess or zero-investment return $R$, the OLS alpha functional is
		\begin{equation}
			\alpha_m(R)
			=
			E[R]
			-
			E[f_m]'
			\operatorname{Var}(f_m)^{-1}
			\operatorname{Cov}(f_m,R).
			\label{eq:alpha_functional_general}
		\end{equation}
		Under \Cref{ass:axis_regular}, this is a bounded linear functional. Necessity follows because the generators belong to $\mathcal{V}_x$. For sufficiency, \eqref{eq:restricted_axis_condition} sets the functional to zero on every generator. Linearity extends the result to finite linear combinations, and continuity extends it to their closed span.
	\end{proof}
	
	The proposition is a population statement. Because $C^x_t(p)-pR^f_t=D^x_t(p)+p\widetilde R^{A,x}_t$, the aggregate gate and zero bridge curve jointly price every prefix excess return; differences of prefixes generate interval portfolios, and \eqref{eq:axis_tail_bridge} generates tails. The restriction is therefore complete within $\mathcal{V}_x$, but says nothing about industry, other-characteristic, or idiosyncratic directions outside that space. This is a coordinate-specific specification requirement, not a statement about the size or mean--variance quality of the factor span. Expanding the attainable mean--variance frontier does not by itself imply that the pricing-error curve is zero on $\mathcal{V}_x$.
	
	\begin{remark}[Measure choice and the canonical implementation]
		\label{rem:measure_invariance}
		The proposition holds for any admissible measure, but the resulting axis depends on that measure. The characteristic determines the order; the measure determines traversal, aggregate weighting, and rank-area portfolio weights. Cap weighting is canonical because it produces the value-weighted aggregate of the valid universe and approaches the investible market interpretation as coverage approaches one. Equal weighting is another admissible axis, not an invariant version of the cap-weighted test. The empirical analysis fixes $\rho=1$ throughout.
	\end{remark}
	
	\begin{remark}[Interpreting aggregate-gate failure]
		\label{rem:gate_failure}
		If $\alpha_m(\widetilde R^{A,x})\neq0$, a zero bridge curve prices only
		$\overline{\operatorname{span}}\{D^x(p):p\in[0,1]\}$,
		not the full space $\mathcal{V}_x$. For accounting axes, aggregate alpha may partly reflect selection into the valid CRSP--Compustat universe. Results for gate-failing models are therefore interpreted conditionally as statements about the internal zero-investment bridge subspace.
	\end{remark}
	
	\subsection{Finite-grid implementation, inference, and calibration}
	\label{subsec:finite_axis_implementation}
	
	The population restriction ranges over all $p\in[0,1]$. I approximate it using a finite stock universe and a prespecified grid $P=\{p_1,\ldots,p_L\}$, so non-rejection means that no distortion is detected at the implemented grid resolution.
	
	The empirical analysis uses the cap measure. Let $w^x_{i,t}$ be the start-of-day buy-and-hold wealth share of stock $i$ on holding day $t$, with
	$\sum_{i=1}^{N^x_t}w^x_{i,t}=1$. Stocks retain their formation-date characteristic order. Define
	\begin{equation}
		c^x_{i,t}
		=
		\sum_{j=1}^iw^x_{j,t},
		\qquad
		\mu^x_{i,t}
		=
		c^x_{i,t}
		-
		\frac12w^x_{i,t},
		\label{eq:finite_axis_midpoint}
	\end{equation}
	and for $p\in(0,1]$ let $k_t(p)=\min\{i:c^x_{i,t}\geq p\}$ and $s_t(p)=c^x_{k_t(p),t}$, with $k_t(0)=s_t(0)=0$. The whole-stock bridge and midpoint rank-area return are
	\begin{equation}
		D^x_t(p)
		=
		\sum_{i=1}^{k_t(p)}
		w^x_{i,t}R_{i,t}
		-
		s_t(p)R^{A,x}_t,
		\qquad
		Y^x_t
		=
		\sum_{i=1}^{N^x_t}
		w^x_{i,t}R_{i,t}
		\left(
		\frac12-\mu^x_{i,t}
		\right).
		\label{eq:finite_axis_bridge}
	\end{equation}
	The bridge longs a whole-stock prefix and shorts the same realized aggregate exposure. Because the boundary stock is not split, $s_t(p)$ may exceed $p$, although closure at both endpoints remains exact. The grid-integrated bridge and midpoint rank-area return therefore need not coincide exactly. I test $SA$ through the rank-area alpha and report the difference as whole-stock boundary approximation error. Across all six models and four axes, the absolute gap never exceeds $0.63$ bp per year, with a median of $0.39$ bp, two orders of magnitude below the distortions studied here.
	
	At each grid point I estimate \eqref{eq:axis_bridge_regression} and collect the intercepts in
	$\widehat{\bm{\alpha}}^x_m$, whose $\ell$th entry is $\widehat\alpha^x_m(p_\ell)$. The empirical integration weights and grid maximum produce the finite-grid versions of $IAE$, $ISE$, and $SUP$.
	
	Let $\bm q_{m,t}=(1,f_{m,t}')'$, $Q_m=E[\bm q_{m,t}\bm q_{m,t}']$, and let $\bm e_1$ select the intercept. With residual vector
	$\bm\varepsilon^x_{m,t}=(\varepsilon^x_{m,t}(p_1),\ldots,\varepsilon^x_{m,t}(p_L))'$,
	the influence vector for the intercept grid is
	$\bm\psi^x_{m,t}=(\bm e_1'Q_m^{-1}\bm q_{m,t})\bm\varepsilon^x_{m,t}$.
	Under standard weak-dependence and moment conditions,
	\begin{equation}
		\sqrt T
		\left(
		\widehat{\bm\alpha}^x_m
		-
		\bm\alpha^x_m
		\right)
		\Rightarrow
		N(0,\Omega^x_m),
		\qquad
		\Omega^x_m
		=
		\sum_{h=-\infty}^{\infty}
		E\!\left[
		\bm\psi^x_{m,t}
		\bm\psi^{x\prime}_{m,t-h}
		\right],
		\label{eq:axis_alpha_influence}
	\end{equation}
	so the asymptotic covariance of the intercept estimator is $\Omega^x_m/T$.\footnote{The only inverse is $Q_m^{-1}$, the second-moment matrix of an intercept and at most six factors. The grid covariance matrix $\Omega^x_m$ is never inverted and may be singular because the endpoint bridges are identically zero.}
	
	I test $SA$ using the HAC $t$-test for the rank-area alpha. For $IAE$, $ISE$, and $SUP$, I estimate the long-run covariance of the influence vectors by Newey--West HAC, impose positive semidefiniteness to obtain $\widehat\Omega^x_{m,+}$, and simulate
	\begin{equation}
		\bm A^{x,(b)}_m
		\sim
		N\!\left(
		\bm0,
		\widehat\Omega^x_{m,+}/T
		\right),
		\qquad
		b=1,\ldots,B.
		\label{eq:axis_hac_gp_draw}
	\end{equation}
	Applying the same grid operations to each draw gives the null distribution of each nonlinear functional, with upper-tail p-value
	\begin{equation}
		\widehat p_F
		=
		\frac{
			1+
			\sum_{b=1}^B
			\mathbf 1
			\left\{
			F(\bm A^{x,(b)}_m)
			\geq
			F(\widehat{\bm\alpha}^x_m)
			\right\}
		}{
			B+1
		},
		\qquad
		F\in\{IAE,ISE,SUP\}.
		\label{eq:axis_hac_gp_pvalue}
	\end{equation}
	With $B=50{,}000$, the minimum attainable p-value is
	$1/(B+1)=2\times10^{-5}$; entries reported as $<0.0001$ attain this bound. The procedure approximates the sampling distribution of the intercept estimator, with serial and cross-grid dependence entering through the HAC covariance. Critical values are model-specific because they use each model's residual covariance; \Cref{subsec:robust_size_power} evaluates the resulting resolution differences in finite samples. The p-values are model--axis-specific, and I do not combine axes into a joint null.
	
	For finite-sample size and power, I hold the observed monthly factor matrix, fitted loadings, and centered residual curves fixed and simulate the zero-curve null and empirical-shape alternatives
	$\bm\alpha^x_m(\lambda)=\lambda\widehat{\bm\alpha}^x_m$,
	$\lambda\in\{0,0.25,0.5,0.75,1\}$. The case $\lambda=0$ evaluates size; positive $\lambda$ preserves the estimated curve's location and shape while scaling its magnitude. \Cref{subsec:robust_size_power} reports rejection frequencies under HAC-GP and block-bootstrap critical values.
	
	\subsection{From maximum-Sharpe gains to axis destination}
	\label{subsec:axis_sr_relation}
	
	The population restriction prices only $\mathcal V_x$, while an empirical pass means only that the finite-grid tests detect no pricing-error structure at the implemented resolution.
	
	Consider adding candidate factor $g$ to base model $m$. Let
	\begin{equation}
		g_t
		=
		\alpha_m(g)
		+
		b_g'f_{m,t}
		+
		\varepsilon_{g,t},
		\qquad
		E[\varepsilon_g]=0,
		\qquad
		\operatorname{Cov}(\varepsilon_g,f_m)=0,
		\label{eq:factor_residualization}
	\end{equation}
	and define the signed appraisal ratio
	\begin{equation}
		IR_g
		=
		\frac{\alpha_m(g)}{\sigma(\varepsilon_g)}.
		\label{eq:signed_appraisal_ratio}
	\end{equation}
	The factor's contribution to the squared maximum Sharpe ratio is
	\begin{equation}
		SR^2(m\cup g)-SR^2(m)
		=
		IR_g^2,
		\label{eq:sr2_decomposition}
	\end{equation}
	so frontier expansion depends on $|IR_g|$, not its sign or the economic coordinate with which the residual factor aligns.
	
	By Frisch--Waugh--Lovell, adding $g$ changes the alpha of any test return $R$ according to
	\begin{equation}
		\alpha_{m\cup g}(R)
		=
		\alpha_m(R)
		-
		\beta_{R,\varepsilon_g}\alpha_m(g),
		\qquad
		\beta_{R,\varepsilon_g}
		=
		\frac{\operatorname{Cov}(R,\varepsilon_g)}
		{\operatorname{Var}(\varepsilon_g)}.
		\label{eq:fwl_alpha_shift}
	\end{equation}
	\footnote{\citet{FrischWaugh33} and \citet{Lovell63}; see also \citet{Lovell08}. In asset pricing, the same algebra underlies the appraisal-ratio results of \citet{GRS89} and \citet{BS17}.}
	
	Applying \eqref{eq:fwl_alpha_shift} pointwise to the bridge gives
	\begin{equation}
		\Delta\alpha^x(p;g\mid m)
		\equiv
		\alpha^x_{m\cup g}(p)
		-
		\alpha^x_m(p)
		=
		-
		\operatorname{corr}
		\left(
		D^x_{\perp m}(p),
		\varepsilon_g
		\right)
		\sigma\!\left(D^x_{\perp m}(p)\right)
		IR_g,
		\label{eq:curve_shift}
	\end{equation}
	where $D^x_{\perp m}(p)$ is the residual from projecting the bridge at cutoff $p$ on the base-model factors and an intercept. The enlarged model leaves the axis at
	\begin{equation}
		\alpha^x_{m\cup g}(p)
		=
		\alpha^x_m(p)
		+
		\Delta\alpha^x(p;g\mid m).
		\label{eq:axis_destination}
	\end{equation}
	Equation \eqref{eq:curve_shift} separates frontier expansion from axis destination: $IR_g$ scales the potential correction, residual alignment determines its cutoff-specific sign and transmission, and the starting curve determines whether the enlarged model undercorrects, flattens, or overshoots the axis.
	
	Integrating over the axis, or applying \eqref{eq:fwl_alpha_shift} directly to the rank-area return, gives
	\begin{equation}
		\Delta SA^x(g\mid m)
		\equiv
		SA^x_{m\cup g}
		-
		SA^x_m
		=
		-
		\operatorname{corr}
		\left(
		Y^x_{\perp m},
		\varepsilon_g
		\right)
		\sigma(Y^x_{\perp m})
		IR_g,
		\label{eq:alpha_shift}
	\end{equation}
	where $Y^x_{\perp m}$ is the residual from projecting $Y^x$ on the base-model factors and an intercept. Equation \eqref{eq:alpha_shift} is the scalar counterpart of the full curve shift, summarizing its average signed movement without retaining the cutoff-specific shape.
	
	Equations \eqref{eq:curve_shift} and \eqref{eq:alpha_shift} are accounting identities, not independent predictions. Their empirical content lies in whether candidate factors and alternative constructions materially change the appraisal ratio, residual alignment, and resulting destination---objects that maximum-Sharpe spanning alone does not recover.
	
	Two properties guide the empirical analysis. First, the signed shift is invariant to factor rescaling: multiplying $g$ by a nonzero constant changes $IR_g$ and the residual correlation in offsetting ways, leaving the alpha shift unchanged. Axis distortion therefore cannot be reduced by reporting the same factor in smaller units. Second, both $IR_g$ and residual alignment belong to the traded portfolio rather than the characteristic in the abstract, so portfolios formed from the same signal can produce different shifts and destinations. \Cref{subsec:construction_sensitivity} tests this distinction by repackaging the same characteristic information. In the empirical tables, $\Delta SA$ follows the signed convention
	\begin{equation}
		SA^x_{m\cup g}
		=
		SA^x_m
		+
		\Delta SA^x(g\mid m).
		\label{eq:sa_destination}
	\end{equation}
	
	\section{Data and Methodology}
	\label{sec:data}
	
	I use daily CRSP stock returns \citep{CRSP}, Compustat annual fundamentals \citep{compustat}, the Fama--French Data Library \citep{FrenchDataLibrary}, q-factor provider data \citep{globalqFactors}, and Global Factor Data \citep{GFD}. The construction proceeds through three nested universes: a month-end capitalization and liquidity screen defines the CRSP investible universe; a daily capitalization screen produces the common clean market and internal return $cMKT$; and each formation date restricts that market to stocks with a valid value, profitability, investment, or momentum signal.
	
	For accounting axes formed in June of year $Y$, I require $fyear=Y-1$ exactly; investment additionally requires $fyear=Y-2$. Firms with older fundamentals are dropped rather than carried forward, and the same convention governs the factor replications in \Cref{subsec:robust_factor_replication}. All tests cover January 1967--December 2024 $(N=696)$, but formation inputs are staged earlier when needed: the 1967 accounting portfolios inherit June 1966 formations, and January 1967 momentum uses returns beginning in January 1966. No reported regression includes returns before 1967.
	
	\subsection{CRSP investible universe and production clean market}
	\label{subsec:crsp-universe}
	
	The base sample contains NYSE, AMEX, and NASDAQ common stocks from January 3, 1967 through December 31, 2024. At each month-end, capitalization and average dollar volume over the previous 63 trading days determine investibility, with asymmetric entry and exit thresholds to reduce boundary turnover.\footnote{Market capitalization is $ME=|PRC|\times SHROUT\times1{,}000$. An incumbent exits if its cumulative capitalization share exceeds 99.9\% or liquidity falls below the monthly 2.5th percentile; an outsider enters only within the 99.5\% capitalization share and at or above the fifth liquidity percentile. The state applies from the next month's first trading day. Early NASDAQ observations are exempt from the liquidity condition until 63 trading days are available. Adjusted returns compound ordinary and delisting returns when both exist and use the delisting return alone when only it is observed.} The screen retains 77.6\% of common stocks on average while preserving approximately 99.7\% of market capitalization and dollar volume. At year-end 2024 it contains 2,426 of 3,804 stocks and preserves 99.78\% of capitalization and 99.28\% of dollar volume.
	
	Within this universe, let $\mathcal{E}_t$ contain stocks with positive lagged capitalization and an observed adjusted return. Ordering them by lagged capitalization, the production clean market $\mathcal{C}_t$ is the smallest whole-stock prefix reaching $\chi=0.9975$:
	\begin{equation}
		\frac{\sum_{i\in\mathcal{C}_t}ME_{i,t-1}}
		{\sum_{i\in\mathcal{E}_t}ME_{i,t-1}}
		\geq\chi,
		\qquad
		cMKT_t
		=
		\frac{\sum_{i\in\mathcal{C}_t}ME_{i,t-1}R^{adj}_{i,t}}
		{\sum_{i\in\mathcal{C}_t}ME_{i,t-1}}.
		\label{eq:production_clean_market}
	\end{equation}
	The boundary stock is retained in full. This characteristic-neutral screen supplies the common formation mask and an internal market benchmark; it is distinct from the buy-and-hold characteristic-valid aggregate used in each bridge. $cMKT$ correlates 0.999778 with the Fama--French market and 0.999799 with the q5 market, with daily RMSEs of 2.21 and 2.10 bp. Its candidate-model regressions have $R^2>0.999$ and insignificant Newey--West 21-day alphas.
	
	The benchmark models are CAPM, FF3, Carhart, FF5, FF6, and q5. FF6 adds momentum to FF5, while q5 contains market, size, investment, return-on-equity, and expected-growth factors. Baseline tests use published monthly factors, aligning the testing horizon with factor construction and reducing short-horizon nonsynchronous-trading concerns.\footnote{Daily tests preserve the main sign pattern, but Dimson corrections with $k=1,\ldots,21$ move the curves substantially and non-monotonically, attenuating negative counterpart-factor curves and enlarging positive curves in models without the counterpart factor. No single lead--lag specification provides a stable daily benchmark.}
	
	\subsection{Reconstructed counterpart factors}
	\label{subsec:prime-factors}
	
	To vary factor packaging while preserving the characteristic signal, I construct univariate counterparts from Fama--French one-dimensional value-weighted portfolios. For book-to-market, operating profitability, and investment, the reconstructed factors are the published $30/70$ spreads
	\begin{equation}
		\begin{aligned}
			\mathrm{HML}'_t
			&=
			\mathrm{Hi30}^{BM}_t
			-
			\mathrm{Lo30}^{BM}_t,
			\\
			\mathrm{RMW}'_t
			&=
			\mathrm{Hi30}^{OP}_t
			-
			\mathrm{Lo30}^{OP}_t,
			\\
			\mathrm{CMA}'_t
			&=
			\mathrm{Lo30}^{INV}_t
			-
			\mathrm{Hi30}^{INV}_t,
		\end{aligned}
		\label{eq:prime_factors}
	\end{equation}
	with investment signed conservative-minus-aggressive. Because prior-12--2 momentum is distributed as deciles, I define
	\begin{equation}
		\mathrm{UMD}'_t
		=
		\frac{1}{3}
		\left(
		D_{8,t}+D_{9,t}+D_{10,t}
		\right)
		-
		\frac{1}{3}
		\left(
		D_{1,t}+D_{2,t}+D_{3,t}
		\right).
		\label{eq:umd_prime}
	\end{equation}
	This compares the upper and lower 30\% of the momentum order, although the decile returns are formed separately and then equally averaged.
	
	Each reconstruction retains the sorting variable, breakpoint convention, value weighting, and formation schedule, but removes the explicit small--big split and cross-size averaging of the canonical $2\times3$ factor. It therefore changes several portfolio-design features simultaneously and is interpreted as alternative packaging, not as a clean causal test of size neutralization. Replacement specifications include either the canonical factor or its reconstruction, never both.
	
	\subsection{Valid-characteristic subuniverse and aggregate gate}
	\label{subsec:valid-subuniverse}
	
	At formation date $\tau$, axis $x$ restricts the production clean market to stocks with positive market equity and a valid characteristic:
	\begin{equation}
		\mathcal{I}^{x}_{\tau}
		=
		\left\{
		i\in\mathcal{C}_{\tau}:
		ME_{i,\tau}>0,\
		x_{i,\tau}\text{ is valid}
		\right\},
		\qquad
		R^{A,x}_t
		=
		\sum_{i\in\mathcal{I}^{x}_{\tau}}
		w^x_{i,t}R^{adj}_{i,t}.
		\label{eq:characteristic_valid_universe_data}
	\end{equation}
	Formation-date capitalization is renormalized within $\mathcal{I}^{x}_{\tau}$, so stocks without valid signals enter neither the rank path nor its aggregate short leg. The resulting $R^{A,x}_t$ is an axis-specific buy-and-hold aggregate and generally differs from daily reconstructed $cMKT_t$, especially for accounting axes with incomplete Compustat coverage.
	
	Before testing the bridge curve, I estimate
	\begin{equation}
		R^{A,x}_t-R^f_t
		=
		a^{A,x}_m
		+
		b^{A,x\prime}_m f_{m,t}
		+
		e^{A,x}_{m,t}.
		\label{eq:aggregate_gate}
	\end{equation}
	The $cMKT$ regressions validate the common market implementation; this aggregate gate asks whether the model prices the characteristic-valid universe itself. If it fails, bridge results are interpreted conditionally on the internal zero-investment subspace.
	
	\subsection{Bridge-path implementation}
	\label{subsec:bridge-path-implementation}
	
	At formation, the engine applies characteristic validity to the clean market and sorts surviving stocks by the predetermined signal, breaking ties by descending market equity and ascending PERMNO. The universe and order remain fixed through the holding period: July--June for annual accounting axes and the next calendar month for momentum. Prefix membership at a fixed mass cutoff may still change as buy-and-hold wealth shares evolve.\footnote{Portfolio accounting uses total and ex-dividend returns. Dividend cash is reinvested across surviving positions in proportion to positive ex-dividend values; positions without usable total returns are removed and remaining wealth is renormalized.}
	
	The baseline path uses 201 equally spaced cutoffs on $[0,1]$. At each date, the first whole-stock prefix reaching target $p$ is selected and the same realized aggregate exposure $s_t(p)$ is shorted. Every run verifies endpoint closure, body--tail recombination, and agreement between the integrated bridge and midpoint rank-area alpha; the latter differ by at most 0.63 bp per year across all models and axes.
	
	Daily bridge and rank-area payoffs are accumulated within each month:
	\begin{equation}
		D^x_M(p)
		=
		\sum_{t\in\mathcal D(M)}D^x_t(p),
		\qquad
		Y^x_M
		=
		\sum_{t\in\mathcal D(M)}Y^x_t.
		\label{eq:monthly_bridge_aggregation}
	\end{equation}
	Arithmetic accumulation is appropriate because these are zero-investment payoff increments. The positive-wealth aggregate is instead compounded:
	\begin{equation}
		R^{A,x}_M
		=
		\prod_{t\in\mathcal D(M)}
		\left(1+R^{A,x}_t\right)-1.
		\label{eq:monthly_aggregate_compounding}
	\end{equation}
	Monthly bridge regressions then use the corresponding published monthly factors.
	
	\subsection{Empirical comparison and factor scan}
	\label{subsec:empirical-design}
	
	For each model--axis pair, I report the rank-area alpha, $IAE$, $\sqrt{ISE}$, and $SUP$, together with their p-values, model-specific 5\% critical values, and $R^2_Y$. The candidate-factor scan adds each of 155 monthly factors separately to CAPM and records its maximum-Sharpe gain and resulting axis distortion; the reconstructed factors are excluded from this library. Finally, the replacement experiment substitutes each reconstructed counterpart factor while holding the tested axis and remaining model fixed. These exercises distinguish frontier expansion, coordinate-specific correction, and sensitivity to traded factor implementation.
	
	\section{Empirical Results}
	\label{sec:result}
	
	This section applies the monthly characteristic-axis diagnostic to value, operating profitability, investment, and momentum, corresponding to HML, RMW/ROE, CMA/IA, and UMD; the cap axis is studied in \citet{Shin26}. The baseline tables report full-sample monthly results, while \Cref{subsec:subperiod} examines whether the headline value destination is stable across calendar and rolling subperiods. \Cref{tab:aggregate-gate-all} reports the valid-characteristic aggregate gates. CAPM, FF5, FF6, and q5 pass all four; FF3 and Carhart fail selected accounting-axis gates, whereas every momentum gate passes. As in \Cref{rem:gate_failure}, bridge results for gate-failing models are interpreted conditionally on the internal zero-investment subspace.
	
	\begin{table}[!t]
		\centering
		\singlespacing
		\caption{Aggregate-gate regressions across the four characteristic axes}
		\label{tab:aggregate-gate-all}
		\footnotesize
		\setlength{\tabcolsep}{4.0pt}
		\begin{tabular}{lrrrrrrrrrrrr}
			\toprule
			& \multicolumn{3}{c}{Value} & \multicolumn{3}{c}{Operating prof.} & \multicolumn{3}{c}{Investment} & \multicolumn{3}{c}{Momentum} \\
			\cmidrule(lr){2-4}\cmidrule(lr){5-7}\cmidrule(lr){8-10}\cmidrule(lr){11-13}
			Model & $\alpha$ & $t$ & $p$ & $\alpha$ & $t$ & $p$ & $\alpha$ & $t$ & $p$ & $\alpha$ & $t$ & $p$ \\
			\midrule
			CAPM    &  20.01 &  1.51 & 0.130 &  19.68 &  1.49 & 0.137 &  23.66 &  1.57 & 0.116 &  8.92 &  1.20 & 0.229 \\
			FF3     &  21.53 &  1.96 & 0.050 &  21.86 &  1.99 & 0.047 &  21.80 &  1.74 & 0.082 &  5.28 &  0.85 & 0.396 \\
			Carhart &  29.01 &  2.33 & 0.020 &  29.71 &  2.38 & 0.017 &  29.07 &  2.07 & 0.038 & 11.22 &  1.65 & 0.099 \\
			FF5     &   5.11 &  0.45 & 0.652 &   6.13 &  0.54 & 0.591 &  -1.84 & -0.15 & 0.879 & -5.54 & -0.83 & 0.408 \\
			FF6     &  12.31 &  1.06 & 0.288 &  13.57 &  1.16 & 0.244 &   5.77 &  0.47 & 0.642 &  0.06 &  0.01 & 0.992 \\
			q5      &  -6.48 & -0.50 & 0.619 &  -4.23 & -0.33 & 0.743 & -15.24 & -1.07 & 0.284 & -0.73 & -0.10 & 0.919 \\
			\bottomrule
		\end{tabular}
		
		\vspace{0.4em}
		\begin{minipage}{0.94\textwidth}
			\footnotesize
			\emph{Note:} Monthly regressions of each characteristic-valid aggregate return on each factor model; $N=696$ throughout. $\alpha$ is the monthly intercept multiplied by 12 and reported in annualized basis points. Tests use Newey--West standard errors with six monthly lags. Regression $R^2$ lies between 0.994 and 0.999. Correlations with the full investible market are 0.9968 on value, 0.9968 on operating profitability, 0.9967 on investment, and 0.9997 on momentum.
		\end{minipage}
	\end{table}
	
	\Cref{fig:axis-fingerprint} displays the four full-sample curve families, and the accompanying tables report signed area, nonlinear functionals, model-specific critical values, and $R^2_Y$. Counterpart-factor content generally shifts the curves downward; the central questions are where the shift lands and how precisely that destination is estimated. Cross-model threshold comparisons are descriptive because each null distribution depends on the model's full residual covariance curve.
	
	\begin{figure}[!t]
		\centering
		\begin{minipage}{0.48\textwidth}
			\centering
			\includegraphics[width=\linewidth]{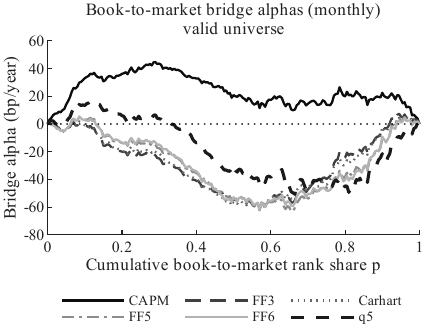}
			\caption*{A. Value}
		\end{minipage}
		\hfill
		\begin{minipage}{0.48\textwidth}
			\centering
			\includegraphics[width=\linewidth]{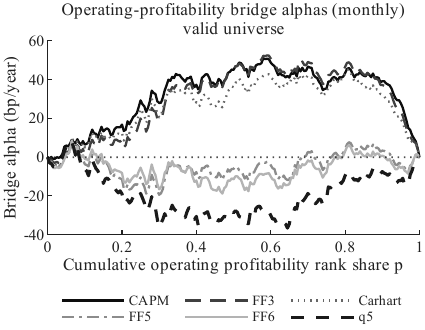}
			\caption*{B. Operating profitability}
		\end{minipage}
		
		\vspace{0.3em}
		
		\begin{minipage}{0.48\textwidth}
			\centering
			\includegraphics[width=\linewidth]{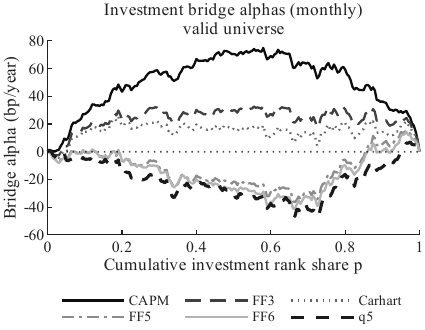}
			\caption*{C. Investment}
		\end{minipage}
		\hfill
		\begin{minipage}{0.48\textwidth}
			\centering
			\includegraphics[width=\linewidth]{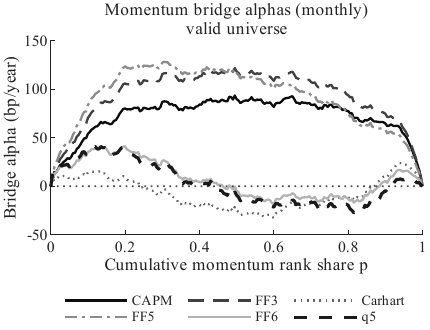}
			\caption*{D. Momentum}
		\end{minipage}
		
		\caption{Bridge-alpha curves on the four characteristic axes}
		\label{fig:axis-fingerprint}
		
		\vspace{0.3em}
		\begin{minipage}{0.94\linewidth}
			\footnotesize
			\emph{Note:} Monthly bridge-alpha curves within each valid-characteristic universe. The horizontal axis is cumulative cap-weight share along the characteristic order; alphas are annualized basis points. \textbf{Vertical scales differ across panels}; cross-axis magnitudes should be compared using the functionals in \Cref{tab:value-axis-main}--\Cref{tab:mom-axis-main}.
		\end{minipage}
	\end{figure}
	
	\FloatBarrier
	
	
	\subsection{Value}
	\label{subsec:value_results}
	
	The value axis sorts stocks each June from high to low by
	\[
	x^{BM}_{i,Y}=\log\left(\frac{BE_{i,Y-1}}{ME^{Dec}_{i,Y-1}}\right),
	\]
	requiring positive book equity and December market equity, and holds the order from July through the following June. The valid universe covers 81.2\% of investible capitalization on average and 87.9\% at the final formation date. Carhart fails the aggregate gate, while FF3 lies at the 5\% boundary.
	
	\begin{table}[!htbp]
		\centering
		\singlespacing
		\caption{Value-axis bridge-alpha diagnostics: monthly baseline}
		\label{tab:value-axis-main}
		\footnotesize
		\setlength{\tabcolsep}{4pt}
		\begin{tabular}{l rrr rrr rrr rrr}
			\toprule
			& \multicolumn{3}{c}{Rank area}
			& \multicolumn{3}{c}{$IAE$}
			& \multicolumn{3}{c}{$\sqrt{ISE}$}
			& \multicolumn{3}{c}{$SUP$} \\
			\cmidrule(lr){2-4}\cmidrule(lr){5-7}\cmidrule(lr){8-10}\cmidrule(lr){11-13}
			Model
			& Rank $\alpha$ & $p$ & $R^2_Y$
			& Stat & $q_{95}$ & $p$
			& Stat & $q_{95}$ & $p$
			& Stat & $q_{95}$ & $p$ \\
			\midrule
			CAPM    &  24.3\nostar          & 0.3226 & 0.029 & 24.1\nostar          & 47.8 & 0.3205 & 26.3\nostar          & 52.9 & 0.3425 & 44.7\nostar          & 79.6 & 0.3365 \\
			FF3     & -27.8$^{\ast\ast\ast}$ & 0.0050 & 0.743 & 28.2$^{\ast\ast\ast}$ & 19.7 & 0.0041 & 34.3$^{\ast\ast\ast}$ & 23.8 & 0.0032 & 61.0$^{\ast\ast\ast}$ & 44.6 & 0.0034 \\
			Carhart & -28.0$^{\ast\ast\ast}$ & 0.0037 & 0.743 & 28.3$^{\ast\ast\ast}$ & 19.2 & 0.0029 & 34.5$^{\ast\ast\ast}$ & 23.1 & 0.0024 & 61.7$^{\ast\ast\ast}$ & 44.2 & 0.0023 \\
			FF5     & -28.6$^{\ast\ast\ast}$ & 0.0030 & 0.760 & 28.8$^{\ast\ast\ast}$ & 19.0 & 0.0032 & 35.4$^{\ast\ast\ast}$ & 23.0 & 0.0023 & 63.2$^{\ast\ast\ast}$ & 44.4 & 0.0027 \\
			FF6     & -28.1$^{\ast\ast\ast}$ & 0.0030 & 0.760 & 28.3$^{\ast\ast\ast}$ & 18.7 & 0.0025 & 35.0$^{\ast\ast\ast}$ & 22.6 & 0.0019 & 61.8$^{\ast\ast\ast}$ & 43.7 & 0.0023 \\
			q5      & -17.2\nostar          & 0.3033 & 0.465 & 22.2\nostar          & 32.6 & 0.1811 & 27.6\nostar          & 36.7 & 0.1474 & 50.7\nostar          & 60.3 & 0.1223 \\
			\bottomrule
		\end{tabular}
		
		\vspace{0.4em}
		\begin{minipage}{0.94\textwidth}
			\footnotesize
			\emph{Note:} Rank $\alpha$ is the annualized alpha of the rank-area portfolio and equals the signed area up to whole-stock boundary approximation error; $R^2_Y$ is its regression $R^2$. The nonlinear functionals and $q_{95}$ values are annualized basis points. $\sqrt{ISE}$ is reported in the units of the alpha curve; the monotone transformation leaves p-values unchanged. Stars are assigned from each statistic's own p-value: $^{\ast\ast\ast}$, $^{\ast\ast}$, and $^{\ast}$ denote significance at the 1\%, 5\%, and 10\% levels. Tests use Newey--West standard errors with six monthly lags. \Cref{tab:op-axis-main}--\Cref{tab:mom-axis-main} use the same definitions.
		\end{minipage}
	\end{table}
	
	Panel A of \Cref{fig:axis-fingerprint} and \Cref{tab:value-axis-main} show a full-sample sign reversal. CAPM leaves an insignificant positive rank-area alpha of $+24.3$ bp, whereas FF3, Carhart, FF5, and FF6 leave significantly negative rank alphas near $-28$ bp and reject every nonlinear functional. In the full sample, HML-based models therefore push the value curve through zero rather than flattening it. Their high $R^2_Y$ values of 0.743--0.760 show that covariance fit and elimination of mean pricing errors are distinct requirements. \Cref{subsec:subperiod} shows that this negative destination is prevalent but not constant, averaging periods of near-zero correction, pronounced late-1990s overshooting, and renewed negative distortion near the sample end.
	
	q5 also lies below zero but is not rejected against its wider threshold. Its $IAE$ of 22.2 bp exceeds the 19.0 bp threshold at which FF5 is rejected, so its non-rejection provides weaker evidence of flatness than the tightly estimated FF5 and FF6 results; CAPM's non-rejection is weaker still given its 47.8 bp threshold. The rolling evidence reinforces this distinction: q5 frequently shares the negative direction of the Fama--French estimates but rejects much less often because its sampling uncertainty and functional null bands are wider. These comparisons do not attribute the negative curve to HML alone.\footnote{Replacing HML with a univariate value spread formed from the same book-to-market ordering also leaves the curve significantly negative. Within FF5, removing HML makes the value rank alpha more negative, from roughly $-28$ to $-40.6$ bp; see \Cref{subsec:construction_sensitivity}.}
	
	
	\subsection{Operating Profitability}
	\label{subsec:op_results}
	
	The operating-profitability axis sorts stocks each June from high to low by
	\[
	x^{OP}_{i,Y}
	=
	\frac{
		REVT_{i,Y-1}
		-
		COGS_{i,Y-1}
		-
		XSGA_{i,Y-1}
		-
		XINT_{i,Y-1}
	}{
		BE_{i,Y-1}
	},
	\]
	and holds the order from July through the following June. Revenue must be observed and book equity positive; missing COGS, XSGA, and XINT are set to zero. The valid universe covers 81.3\% of investible capitalization on average and 87.9\% at the final formation date. FF3 and Carhart fail the aggregate gate.
	
	\begin{table}[!htbp]
		\centering
		\singlespacing
		\caption{OP-axis bridge-alpha diagnostics: monthly baseline}
		\label{tab:op-axis-main}
		\footnotesize
		\setlength{\tabcolsep}{4.0pt}
		\begin{tabular}{l r@{}l rr r@{}l rr r@{}l rr r@{}l rr}
			\toprule
			& \multicolumn{4}{c}{Rank area}
			& \multicolumn{4}{c}{$IAE$}
			& \multicolumn{4}{c}{$\sqrt{ISE}$}
			& \multicolumn{4}{c}{$SUP$} \\
			\cmidrule(lr){2-5}\cmidrule(lr){6-9}\cmidrule(lr){10-13}\cmidrule(lr){14-17}
			Model
			& \multicolumn{2}{c}{Rank $\alpha$} & $p$ & $R^2_Y$
			& \multicolumn{2}{c}{Stat} & $q_{95}$ & $p$
			& \multicolumn{2}{c}{Stat} & $q_{95}$ & $p$
			& \multicolumn{2}{c}{Stat} & $q_{95}$ & $p$ \\
			\midrule
			CAPM    &  32.0 & $^{*}$   & 0.0520 & 0.007 & 31.9 & $^{*}$  & 32.5 & 0.0554 & 35.1 & $^{*}$  & 36.2 & 0.0575 & 51.1 &        & 59.2 & 0.1017 \\
			FF3     &  29.5 & $^{**}$  & 0.0495 & 0.162 & 30.1 & $^{**}$ & 29.4 & 0.0449 & 34.4 & $^{**}$ & 33.0 & 0.0405 & 52.6 & $^{*}$ & 55.3 & 0.0677 \\
			Carhart &  26.2 & $^{*}$   & 0.0518 & 0.165 & 26.8 & $^{**}$ & 26.7 & 0.0493 & 29.7 & $^{*}$  & 30.2 & 0.0537 & 45.3 &        & 52.0 & 0.1068 \\
			FF5     &  -5.0 &          & 0.5951 & 0.529 &  6.1 &         & 19.5 & 0.8029 &  7.5 &         & 23.2 & 0.8087 & 19.4 &        & 45.1 & 0.7535 \\
			FF6     &  -6.6 &          & 0.5061 & 0.530 &  7.0 &         & 20.1 & 0.6986 &  8.6 &         & 24.1 & 0.7119 & 18.8 &        & 46.7 & 0.7949 \\
			q5      & -18.6 &          & 0.2091 & 0.296 & 19.5 &         & 29.3 & 0.2049 & 22.2 &         & 33.6 & 0.2215 & 36.6 &        & 59.0 & 0.3345 \\
			\bottomrule
		\end{tabular}
		
		\vspace{0.4em}
		\begin{minipage}{0.94\textwidth}
			\footnotesize
			\emph{Note:} Columns are defined as in \Cref{tab:value-axis-main}. This is the only panel in which the nonlinear functionals disagree at 5\%: Carhart's $IAE$ is rejected while its $\sqrt{ISE}$ and $SUP$ are not.
		\end{minipage}
	\end{table}
	
	CAPM, FF3, and Carhart leave positive OP-axis curves near the 5\% boundary. This is the only axis on which the functional choice changes the 5\% verdict: FF3 is rejected by rank area, $IAE$, and $\sqrt{ISE}$ but not by $SUP$, while Carhart is rejected only by $IAE$. Because FF3 and Carhart fail the aggregate gate, these are conditional bridge-subspace statements.
	
	FF5 and FF6 instead leave rank-area alphas of $-5.0$ and $-6.6$ bp and $IAE$ values of 6.1 and 7.0 bp, far below relatively tight thresholds. Their small point estimates and $R^2_Y$ near 0.53 provide the strongest evidence of flat OP curves. q5 lies farther below zero, with rank alpha $-18.6$ bp and $IAE=19.5$ bp, but remains below its wider 29.3 bp threshold. Its non-rejection is therefore less informative than those of FF5 and FF6.
	
	
	\subsection{Investment}
	\label{subsec:inv_results}
	
	The investment axis computes
	\[
	INV_{i,Y}
	=
	\frac{AT_{i,Y-1}-AT_{i,Y-2}}{AT_{i,Y-2}},
	\]
	and sorts on $-INV$, from conservative to aggressive investment. The construction requires positive total assets in both fiscal years and an exact one-fiscal-year lag. The valid universe covers 79.7\% of investible capitalization on average and 90.8\% at the final formation date. Carhart fails the aggregate gate, while FF3 is borderline at the 10\% level.
	
	\begin{table}[!htbp]
		\centering
		\singlespacing
		\caption{Investment-axis bridge-alpha diagnostics: monthly baseline}
		\label{tab:inv-axis-main}
		\footnotesize
		\setlength{\tabcolsep}{4.0pt}
		\begin{tabular}{l r@{}l rr r@{}l rr r@{}l rr r@{}l rr}
			\toprule
			& \multicolumn{4}{c}{Rank area}
			& \multicolumn{4}{c}{$IAE$}
			& \multicolumn{4}{c}{$\sqrt{ISE}$}
			& \multicolumn{4}{c}{$SUP$} \\
			\cmidrule(lr){2-5}\cmidrule(lr){6-9}\cmidrule(lr){10-13}\cmidrule(lr){14-17}
			Model
			& \multicolumn{2}{c}{Rank $\alpha$} & $p$ & $R^2_Y$
			& \multicolumn{2}{c}{Stat} & $q_{95}$ & $p$
			& \multicolumn{2}{c}{Stat} & $q_{95}$ & $p$
			& \multicolumn{2}{c}{Stat} & $q_{95}$ & $p$ \\
			\midrule
			CAPM    &  50.0 & $^{***}$ & 0.0028 & 0.156 & 49.7 & $^{***}$ & 32.7 & 0.0025 & 53.4 & $^{***}$ & 37.4 & 0.0041 & 74.9 & $^{**}$ & 62.2 & 0.0161 \\
			FF3     &  23.5 & $^{**}$  & 0.0261 & 0.489 & 23.9 & $^{**}$  & 20.9 & 0.0239 & 25.1 & $^{**}$  & 24.3 & 0.0423 & 33.4 &         & 43.7 & 0.1957 \\
			Carhart &  14.8 &          & 0.1613 & 0.502 & 15.2 &          & 20.9 & 0.1636 & 16.0 &          & 24.3 & 0.2347 & 25.2 &         & 44.2 & 0.4747 \\
			FF5     & -12.5 &          & 0.1185 & 0.761 & 14.4 &          & 16.7 & 0.1027 & 17.5 & $^{*}$   & 19.8 & 0.0945 & 36.5 & $^{*}$  & 38.4 & 0.0698 \\
			FF6     & -15.3 & $^{*}$   & 0.0876 & 0.763 & 16.2 & $^{*}$   & 18.1 & 0.0861 & 20.3 & $^{*}$   & 21.6 & 0.0704 & 41.2 & $^{*}$  & 42.0 & 0.0567 \\
			q5      & -19.7 & $^{*}$   & 0.0976 & 0.683 & 20.7 & $^{*}$   & 23.7 & 0.0909 & 24.3 & $^{*}$   & 28.1 & 0.0962 & 47.6 & $^{*}$  & 51.9 & 0.0813 \\
			\bottomrule
		\end{tabular}
		
		\vspace{0.4em}
		\begin{minipage}{0.94\textwidth}
			\footnotesize
			\emph{Note:} Columns are defined as in \Cref{tab:value-axis-main}. \Cref{subsec:robust_block_bootstrap} reports the corresponding block-bootstrap results.
		\end{minipage}
	\end{table}
	
	Investment exhibits the largest directional shift. CAPM leaves a significantly positive rank-area alpha of $+50.0$ bp, FF3 reduces it to $+23.5$ bp but remains rejected, and Carhart leaves an insignificant $+14.8$ bp. FF5, FF6, and q5 instead cross zero, with rank alphas of $-12.5$, $-15.3$, and $-19.7$ bp---a movement of about 70 bp from CAPM to q5.
	
	Despite high $R^2_Y$, none of the investment-factor models rejects at 5\% under HAC-GP, although most statistics lie between the 5\% and 10\% levels. The evidence therefore indicates a clear negative shift but only borderline support that the destination differs from zero. The moving-block bootstrap rejects selected FF6 and q5 functionals in \Cref{subsec:robust_block_bootstrap}.
	
	
	\subsection{Momentum}
	\label{subsec:mom_results}
	
	At each month-end, the momentum axis sorts stocks from winners to losers using
	\[
	x^{MOM}_{i,m}
	=
	\prod_{s=2}^{12}\left(1+R_{i,m-s}\right)-1,
	\]
	and holds the order during the following month. I require at least eight valid returns in the 11-month formation window. The valid universe covers 96.8\% of investible capitalization on average and 99.5\% at the final formation date. All aggregate gates pass.
	
	\begin{table}[!htbp]
		\centering
		\singlespacing
		\caption{Momentum-axis bridge-alpha diagnostics: monthly baseline}
		\label{tab:mom-axis-main}
		\footnotesize
		\setlength{\tabcolsep}{4.0pt}
		\begin{tabular}{l r@{}l rr r@{}l rr r@{}l rr r@{}l rr}
			\toprule
			& \multicolumn{4}{c}{Rank area}
			& \multicolumn{4}{c}{$IAE$}
			& \multicolumn{4}{c}{$\sqrt{ISE}$}
			& \multicolumn{4}{c}{$SUP$} \\
			\cmidrule(lr){2-5}\cmidrule(lr){6-9}\cmidrule(lr){10-13}\cmidrule(lr){14-17}
			Model
			& \multicolumn{2}{c}{Rank $\alpha$} & $p$ & $R^2_Y$
			& \multicolumn{2}{c}{Stat} & $q_{95}$ & $p$
			& \multicolumn{2}{c}{Stat} & $q_{95}$ & $p$
			& \multicolumn{2}{c}{Stat} & $q_{95}$ & $p$ \\
			\midrule
			CAPM    & 72.2 & $^{***}$ & 0.0014 & 0.000 & 72.1 & $^{***}$ & 44.0 & 0.0016 &  74.9 & $^{***}$ & 48.4 & 0.0026 &  93.8 & $^{***}$ & 73.0 & 0.0093 \\
			FF3     & 92.6 & $^{***}$ & 0.0000 & 0.092 & 93.0 & $^{***}$ & 43.2 & 0.0000$^{\dagger}$ &  97.1 & $^{***}$ & 47.4 & 0.0000$^{\dagger}$ & 122.0 & $^{***}$ & 71.2 & 0.0003 \\
			Carhart & -7.1 &          & 0.5374 & 0.794 & 15.0 &          & 23.0 & 0.2371 &  17.0 &          & 26.5 & 0.2598 &  32.5 &          & 46.3 & 0.2556 \\
			FF5     & 91.2 & $^{***}$ & 0.0012 & 0.101 & 91.6 & $^{***}$ & 55.2 & 0.0017 &  96.7 & $^{***}$ & 60.0 & 0.0020 & 128.5 & $^{***}$ & 87.0 & 0.0031 \\
			FF6     &  7.0 &          & 0.5536 & 0.805 & 16.1 &          & 23.4 & 0.1929 &  19.6 &          & 26.8 & 0.1739 &  41.2 & $^{*}$   & 46.2 & 0.0949 \\
			q5      &  3.3 &          & 0.9028 & 0.134 & 17.8 &          & 52.4 & 0.5315 &  20.7 &          & 57.4 & 0.5254 &  41.0 &          & 85.4 & 0.4561 \\
			\bottomrule
		\end{tabular}
		
		\vspace{0.4em}
		\begin{minipage}{0.94\textwidth}
			\footnotesize
			\emph{Note:} Columns are defined as in \Cref{tab:value-axis-main}. $^{\dagger}$~No simulated draw exceeded the observed statistic, so the p-value attains its lower bound $1/(B+1)=2\times10^{-5}$ rather than zero. The FF3 rank-area p-value is $3.3\times10^{-5}$.
		\end{minipage}
	\end{table}
	
	Momentum shows the sharpest separation by counterpart-factor content. CAPM, FF3, and FF5 leave large positive curves and reject every functional. Adding UMD moves the rank alpha from $+92.6$ to $-7.1$ bp between FF3 and Carhart and from $+91.2$ to $+7.0$ bp between FF5 and FF6; neither UMD model rejects at 5\%. Unlike HML on value and CMA or IA on investment, UMD supplies a large correction without significant opposite-signed overshoot.
	
	q5 also leaves a small curve despite containing no UMD: its rank alpha is $+3.3$ bp and its $IAE$ is 17.8 bp. Its low $R^2_Y$ and 52.4 bp threshold make the p-value alone weak evidence, but the point estimate lies below the tighter thresholds attained by Carhart and FF6. Its non-rejection is therefore not solely an artifact of its own loose null.
	
	Because all gates pass and UMD removes approximately 84--100 bp of signed distortion without significant overshoot, momentum serves as a positive control: the diagnostic recovers known factor content and can produce a nearly flat curve when the correction aligns with the coordinate. \Cref{subsec:cross_axis_factor_scan} isolates single-factor contributions, and \Cref{subsec:construction_sensitivity} tests whether the destinations depend on traded implementation.
	
	\subsection{Time variation in value-axis correction}
	\label{subsec:subperiod}
	
	The full-sample statistics establish the economic magnitude of the pricing errors but not whether their destination is stable through time. Because value provides the paper's clearest case of correction passing beyond zero, I examine it as a focused setting rather than claim a common chronology across all axes. I repeat the monthly value-axis analysis over calendar subperiods and 49 complete rolling 120-month windows, re-estimating the aggregate gate and model-specific inference within each sample. Failed-gate cells are reported descriptively, and the overlapping rolling windows trace persistent episodes rather than independent repetitions of the test.
	
	\Cref{tab:value-subperiod-decades} reports four models with distinct roles: CAPM is the uncorrected benchmark, FF3 isolates the introduction of HML, FF5 represents the headline overshoot, and q5 provides a non-HML comparison with different estimation precision. The final row covers only January 2020 through December 2024 and is therefore a partial-period diagnostic rather than a complete decade.
	
	\begin{table}[!htbp]
		\centering
		\singlespacing
		\caption{Value-axis correction across calendar subperiods}
		\label{tab:value-subperiod-decades}
		\footnotesize
		\setlength{\tabcolsep}{3.8pt}
		\begin{tabular}{lrrrrrrrr}
			\toprule
			& \multicolumn{2}{c}{CAPM}
			& \multicolumn{2}{c}{FF3}
			& \multicolumn{2}{c}{FF5}
			& \multicolumn{2}{c}{q5} \\
			\cmidrule(lr){2-3}
			\cmidrule(lr){4-5}
			\cmidrule(lr){6-7}
			\cmidrule(lr){8-9}
			Period
			& Rank $\alpha$ & $p(IAE)$
			& Rank $\alpha$ & $p(IAE)$
			& Rank $\alpha$ & $p(IAE)$
			& Rank $\alpha$ & $p(IAE)$ \\
			\midrule
			1970--1979
			& $\mathbf{+104.5}$ & $\mathbf{0.030}$
			& $-5.8$           & $0.894$
			& $-3.9^{\ddagger}$& $0.903$
			& $+22.5$          & $0.786$ \\
			
			1980--1989
			& $\mathbf{+89.6}$ & $\mathbf{0.003}$
			& $-14.0$          & $0.517$
			& $-18.5$          & $0.427$
			& $-2.0$           & $0.615$ \\
			
			1990--1999
			& $-46.5$          & $0.392$
			& $\mathbf{-85.2}$ & $\mathbf{<0.001}$
			& $\mathbf{-77.5}$ & $\mathbf{<0.001}$
			& $-53.0$          & $0.195$ \\
			
			2000--2009
			& $+116.6$         & $0.058$
			& $+3.8$           & $0.797$
			& $-12.6$          & $0.588$
			& $+19.8$          & $0.703$ \\
			
			2010--2019
			& $-55.7$          & $0.074$
			& $-16.4$          & $0.281$
			& $-17.9$          & $0.145$
			& $-19.5$          & $0.393$ \\
			
			2020--2024
			& $-130.2^{\ddagger}$ & $0.302$
			& $-107.7^{\ddagger}$ & $<0.001$
			& $\mathbf{-94.1}$    & $\mathbf{<0.001}$
			& $-96.1$             & $0.053$ \\
			\bottomrule
		\end{tabular}
		
		\vspace{0.4em}
		\begin{minipage}{0.96\textwidth}
			\footnotesize
			\emph{Note:} Rank $\alpha$ is the annualized signed area of the value-axis bridge-alpha curve, in basis points. $p(IAE)$ is the model-specific HAC-GP p-value for the integrated absolute error, based on 4,999 simulated draws. Bold entries denote 5\% functional rejections in cells that pass the aggregate gate. $^{\ddagger}$~The characteristic-valid aggregate rejects at 5\%, so the corresponding axis statistics retain only the conditional interpretation in \Cref{rem:gate_failure}. The 2020--2024 row contains 60 monthly observations and is not a complete decade.
		\end{minipage}
	\end{table}
	
	The calendar results reveal three correction states. In the 1970s and 1980s, CAPM leaves large positive distortions while the HML-based models move the destination close to zero. In the 1990s, FF3 and FF5 overshoot sharply, with rank-area alphas of $-85.2$ and $-77.5$ bp and strong zero-curve rejections despite passing the aggregate gate. The negative destination attenuates in the 2000s and 2010s, then reappears in 2020--2024: FF5 remains aggregate-valid with a rank alpha of $-94.1$ bp, whereas q5 produces a similar estimate of $-96.1$ bp but an $IAE$ p-value of 0.053. Because this final period contains only 60 observations, complete rolling windows provide the cleaner check on the recent shift.
	
	Calendar bins can conceal changes within decades. \Cref{fig:value-rolling-subperiod} therefore reports all complete rolling 10-year windows from 1967--1976 through 2015--2024. Panel A plots rank-area alphas for CAPM, FF3, FF5, and q5; Panel B compares the HAC standard errors of the FF5 and q5 estimates.
	
	\begin{figure}[H]
		\centering
		
		\begin{minipage}[t]{0.48\textwidth}
			\centering
			\includegraphics[width=\linewidth]{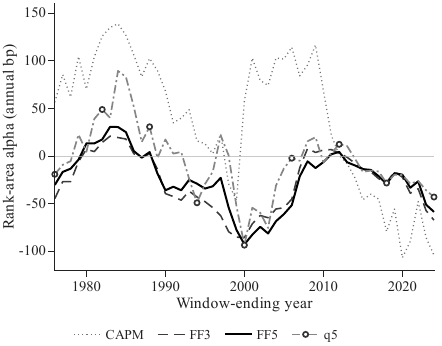}
			\caption*{A. Rolling rank-area alpha}
		\end{minipage}
		\hfill
		\begin{minipage}[t]{0.48\textwidth}
			\centering
			\includegraphics[width=\linewidth]{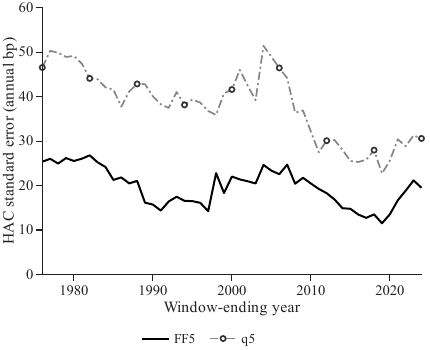}
			\caption*{B. Rolling HAC standard error}
		\end{minipage}
		
		\caption{Rolling ten-year value-axis destinations and estimation precision}
		\label{fig:value-rolling-subperiod}
		
		\vspace{0.3em}
		\begin{minipage}{0.96\linewidth}
			\footnotesize
			\emph{Note:} Panel A reports annualized rank-area alphas for CAPM, FF3, FF5, and q5 over complete 120-month windows, indexed by the ending year. Panel B reports the corresponding HAC standard errors for FF5 and q5. All point estimates are displayed to preserve the time-series chronology; a window that fails the aggregate gate is not treated as unconditional evidence. FF6 is omitted for readability because its rolling rank-area alpha closely tracks FF5, with a correlation of 0.977. Adjacent windows overlap and are not statistically independent. \Cref{app:value_subperiod_full} reports rolling-window summary statistics and selected turning points; the complete machine-readable model-by-window output is included in the replication archive.
		\end{minipage}
	\end{figure}
	
	The rolling evidence shows that the FF5 overshoot is prevalent but not constant. Its rank-area alpha is negative in 39 of 49 windows and in 34 of the 44 windows that pass the aggregate gate. The point estimate nevertheless rises to $+30.7$ bp in 1974--1983 before turning negative again; this positive estimate is insignificant and therefore indicates a shifted destination, not a significant opposite-signed failure. FF5 reaches its trough of $-92.7$ bp in 1991--2000, with $p(IAE)<0.001$, subsequently approaches zero, and becomes negative again near the end of the sample. The complete 2012--2021 and 2015--2024 windows pass the aggregate gate and yield rank alphas of $-33.3$ and $-58.6$ bp, with $p(IAE)=0.046$ and 0.002. The recent evidence is thus a late-sample re-emergence of negative distortion, not a structural break identified at a precise date.
	
	The q5 comparison separates estimated distortion from statistical resolution. Among aggregate-valid windows, q5 produces a negative rank-area estimate in 28 of 47 cases but rejects the $IAE$ null in only three, compared with 34 negative estimates and 16 rejections among 44 valid FF5 windows. Its median rolling HAC standard error is 39.2 bp, versus 20.5 bp for FF5, and its median $IAE$ critical value is 76.3 bp, versus 41.4 bp. q5 therefore often moves in the same negative direction as the Fama--French models, but its wider null distribution produces much weaker rejection evidence; non-rejection does not imply a flatter estimated curve.
	
	The chronology also does not mechanically track Compustat coverage. Average capitalization coverage rises smoothly from 73.8\% in 1974--1983 to 82.2\% in 1991--2000 and 89.2\% in 2015--2024, while the estimated destination moves from positive to strongly negative, back toward zero, and then negative again. Lower early-sample coverage remains a scope qualification, but neither it nor aggregate mismatch explains the sign reversal.
	
	Overall, the full-sample value overshoot averages economically distinct correction states. HML-based models bring the positive CAPM distortion close to zero in some periods, cross into a broad negative curve in others, and display a renewed negative episode near the end of the sample. The axis diagnostic therefore identifies not only whether correction overshoots on average, but when its destination moves toward and across zero. The next subsection returns to the full-sample model-by-axis pricing fingerprint.
	
	\FloatBarrier
	
	\subsection{Economic magnitude under unit gross exposure}
	\label{subsec:economic_magnitude}
	
	The rank-area alpha has a natural portfolio interpretation, but its raw magnitude reflects the exposure embedded in the rank-weight function. From \eqref{eq:axis_rank_area}, the continuous rank-area payoff uses weight density
	\[
	h(u)=\frac12-u,
	\qquad
	u\in[0,1].
	\]
	Its net exposure is zero, while gross exposure equals
	\begin{equation}
		G=\int_0^1\left|\frac12-u\right|du=2\int_0^{1/2}\left(\frac12-u\right)du=\frac14.
		\label{eq:rank_area_continuous_gross}
	\end{equation}
	The portfolio therefore contains one-eighth long and one-eighth short exposure, so expressing the same payoff at unit gross multiplies it by four.
	
	The finite-stock implementation is nearly identical. Using the midpoint ranks in \eqref{eq:finite_axis_midpoint} and the rank-area payoff in \eqref{eq:finite_axis_bridge}, define daily gross exposure as
	\begin{equation}
		G^x_t=\sum_{i=1}^{N^x_t}w^x_{i,t}\left|\frac12-\mu^x_{i,t}\right|.
		\label{eq:finite_rank_area_gross}
	\end{equation}
	This is the midpoint sum for \eqref{eq:rank_area_continuous_gross}. Because $\left|\frac12-u\right|$ is linear except at $u=\frac12$, the approximation is exact on every stock interval except, if present, the one straddling the midpoint. Hence
	\begin{equation}
		G^x_t=\frac14-\delta^x_t,
		\qquad
		0\leq\delta^x_t\leq\frac14\left(w^x_{k_t,t}\right)^2,
		\label{eq:finite_rank_area_gross_error}
	\end{equation}
	where $k_t$ is the stock whose cap-mass interval contains $u=\frac12$. Whole-stock discreteness therefore produces only a second-order correction governed by the midpoint stock's weight.
	
	I normalize the daily rank-area payoff to unit gross,
	\begin{equation}
		Y^{x,UG}_t=\frac{Y^x_t}{G^x_t},
		\label{eq:unit_gross_rank_area}
	\end{equation}
	aggregate it to months as in \eqref{eq:monthly_bridge_aggregation}, and re-estimate the factor-model regression. Exact normalization differs slightly from multiplying the estimated alpha by four because $G^x_t$ varies over time, but the variation is negligible:
	\begin{equation}
		\alpha^{x,UG}_m\approx4\alpha^x_m(Y^x)=4SA^x_m.
		\label{eq:unit_gross_rule_of_four}
	\end{equation}
	
	\Cref{tab:unit-gross-magnitude} reports FF6 as the representative specification because it passes all four aggregate gates and contains the canonical counterpart factor for every tested axis. The table compares the raw rank-area alpha with the alpha obtained by re-estimating the daily unit-gross payoff in \eqref{eq:unit_gross_rank_area}.
	
	\begin{table}[!htbp]
		\centering
		\singlespacing
		\caption{Economic magnitude of FF6 rank-area alphas at unit gross exposure}
		\label{tab:unit-gross-magnitude}
		\footnotesize
		\setlength{\tabcolsep}{6pt}
		\begin{tabular}{lrrrr}
			\toprule
			Axis
			& Mean gross
			& Raw rank $\alpha$
			& Unit-gross $\alpha$
			& Alpha multiple \\
			\midrule
			Value
			& 0.249993
			& $-28.1$
			& $-112.5$
			& 4.000 \\
			Operating profitability
			& 0.249990
			& $-6.6$
			& $-26.5$
			& 4.007 \\
			Investment
			& 0.249963
			& $-15.3$
			& $-61.2$
			& 4.005 \\
			Momentum
			& 0.249981
			& $+7.0$
			& $+28.1$
			& 4.005 \\
			\bottomrule
		\end{tabular}
		
		\vspace{0.4em}
		\begin{minipage}{0.94\textwidth}
			\footnotesize
			\emph{Note:} Alphas are annualized basis points. Mean gross is the time-series mean of $G^x_t$ in \eqref{eq:finite_rank_area_gross}. Unit-gross alphas are obtained by normalizing the daily rank-area payoff by $G^x_t$, aggregating to months, and re-estimating the FF6 regression. Alpha multiple is the unit-gross alpha divided by the raw rank-area alpha. FF6 is reported because it passes every aggregate gate and includes HML, RMW, CMA, and UMD.
		\end{minipage}
	\end{table}
	
	Mean daily gross exposure ranges only from 0.249963 to 0.249993, and the exact alpha multiple ranges from 4.000 to 4.007. On value, unit-gross normalization raises the magnitude of the full-sample FF6 rank alpha from $-28.1$ to $-112.5$ bp per year, implying an economically nontrivial pricing error of approximately 1.1\% per year at 100\% gross exposure. This is a full-sample scale rather than a time-invariant annual error; \Cref{subsec:subperiod} shows that the value destination varies substantially across subperiods.
	
	The remaining axes obey the same scaling. Exact unit-gross normalization moves the operating-profitability alpha from $-6.61$ to $-26.49$ bp, the investment alpha from $-15.27$ to $-61.17$ bp, and the momentum alpha from $+7.03$ to $+28.14$ bp. The largest absolute difference between exact normalization and four times the raw alpha is only 0.07 bp per year. Multiplying the reported rank-area alpha by four is therefore an accurate rule of thumb for economic magnitude. The unit-gross estimates provide scale for the signed-area statistic; they do not replace the raw rank-area alpha, which retains the exact integral interpretation in \eqref{eq:axis_rank_area_alpha}.
	
	\subsection{Cross-axis comparison}
	\label{subsec:summary_results}
	
	\Cref{tab:monthly-axis-fingerprint} summarizes the full-sample monthly rank-area alpha and $IAE$ for the six benchmark models. No model dominates globally: the sign, magnitude, precision, and rejection status of the pricing-error curve vary by coordinate, making the model-by-axis fingerprint the central empirical object.
	
	\begin{table}[!htbp]
		\centering
		\singlespacing
		\caption{Full-sample monthly characteristic-axis fingerprint}
		\label{tab:monthly-axis-fingerprint}
		\footnotesize
		\setlength{\tabcolsep}{4.0pt}
		\begin{tabular}{lrrrrrrrr}
			\toprule
			& \multicolumn{2}{c}{Value}
			& \multicolumn{2}{c}{OP}
			& \multicolumn{2}{c}{Investment}
			& \multicolumn{2}{c}{Momentum} \\
			\cmidrule(lr){2-3}
			\cmidrule(lr){4-5}
			\cmidrule(lr){6-7}
			\cmidrule(lr){8-9}
			Model
			& Rank $\alpha$ & $IAE$
			& Rank $\alpha$ & $IAE$
			& Rank $\alpha$ & $IAE$
			& Rank $\alpha$ & $IAE$ \\
			\midrule
			CAPM
			& $24.3$          & $24.1$
			& $32.0^{*}$      & $31.9^{*}$
			& $50.0^{***}$    & $49.7^{***}$
			& $72.2^{***}$    & $72.1^{***}$ \\
			
			FF3
			& $-27.8^{***}$   & $28.2^{***}$
			& $29.5^{**}$     & $30.1^{**}$
			& $23.5^{**}$     & $23.9^{**}$
			& $92.6^{***}$    & $93.0^{***\dagger}$ \\
			
			Carhart
			& $-28.0^{***}$   & $28.3^{***}$
			& $26.2^{*}$      & $26.8^{**}$
			& $14.8$          & $15.2$
			& $-7.1$          & $15.0$ \\
			
			FF5
			& $-28.6^{***}$   & $28.8^{***}$
			& $-5.0$          & $6.1$
			& $-12.5$         & $14.4^{\ddagger}$
			& $91.2^{***}$    & $91.6^{***}$ \\
			
			FF6
			& $-28.1^{***}$   & $28.3^{***}$
			& $-6.6$          & $7.0$
			& $-15.3^{*}$     & $16.2^{*}$
			& $7.0$           & $16.1^{\ddagger}$ \\
			
			q5
			& $-17.2$         & $22.2$
			& $-18.6$         & $19.5$
			& $-19.7^{*}$     & $20.7^{*}$
			& $3.3$           & $17.8$ \\
			\bottomrule
		\end{tabular}
		
		\vspace{0.4em}
		\begin{minipage}{0.95\textwidth}
			\footnotesize
			\emph{Note:} Rank $\alpha$ is the annualized signed area and $IAE$ the annualized integrated absolute value of the full-sample monthly bridge-alpha curve, both in basis points. Stars are assigned from their respective HAC and HAC-GP p-values: $^{***}$, $^{**}$, and $^{*}$ denote significance at the 1\%, 5\%, and 10\% levels. $^{\dagger}$~No simulated draw exceeded the observed statistic, so the p-value attains its lower bound $1/(B+1)=2\times10^{-5}$. $^{\ddagger}$~Other nonlinear functionals reject at 10\% although $IAE$ does not: $\sqrt{ISE}$ and $SUP$ for FF5 on investment, and $SUP$ for FF6 on momentum. Axis-specific tables report all statistics, critical values, p-values, and $R^2_Y$.
		\end{minipage}
	\end{table}
	
	Across the full-sample model--axis cells, models without the relevant counterpart factor generally leave positive curves, while counterpart-factor content shifts them downward. Where the correction lands differs by coordinate. HML-based models cross into significantly negative value errors; investment-factor models also cross zero, but with only marginal monthly evidence. FF5 and FF6 flatten operating profitability, while Carhart and FF6 flatten momentum. Factor content therefore predicts the direction of correction more reliably than whether the resulting curve is tightly centered at zero.
	
	No benchmark model dominates. FF5 and FF6 produce the flattest profitability curves but are rejected on value, and FF5 leaves a large momentum distortion without UMD. Carhart is comparatively close to zero on investment and momentum but is rejected on value and remains marginal on profitability. In the full-sample monthly tests, q5 is not rejected at 5\% on any axis under HAC-GP, although its value and investment distortions remain economically nontrivial and the moving-block bootstrap rejects selected investment functionals. High covariance fit is likewise insufficient: HML-based models are rejected on value despite high $R^2_Y$, whereas q5 passes momentum with $R^2_Y=0.134$.
	
	The strength of non-rejection depends on model-specific resolution. On momentum, q5's $IAE$ of 17.8 bp lies below FF6's tighter 23.4 bp critical value, so its small estimate is not merely an artifact of a loose null. On profitability, q5's 19.5 bp distortion approximately equals FF5's 19.5 bp threshold. On value and investment, however, q5's $IAE$ values of 22.2 and 20.7 bp exceed the tighter thresholds faced by FF5 and FF6---19.0 bp on value and 16.7--18.1 bp on investment. These descriptive comparisons do not transplant one model's null distribution to another, but they show that q5's value and investment non-rejections provide weaker evidence of flatness. The rolling value results sharpen the point: q5 frequently shares the negative direction of the Fama--French estimates but rejects far less often because its sampling uncertainty and functional null bands are wider. CAPM value is the clearest low-resolution case, with an $IAE$ of 24.1 bp tested against a 47.8 bp threshold.
	
	The aggregate gate leaves the main unconditional findings intact. The FF5 and FF6 value rejections and the positive CAPM investment and momentum rejections occur in clean-gate rows, while the clean-gate CAPM profitability result is marginal. FF3 and Carhart retain the conditional interpretation in \Cref{rem:gate_failure} on the accounting axes. The next subsection isolates individual factor contributions by adding each candidate separately to CAPM.
	
	\FloatBarrier
	
	\subsection{Cross-axis factor scan}
	\label{subsec:cross_axis_factor_scan}
	
	I add each of 155 monthly candidate factors $g$ separately to CAPM and re-estimate all four bridge-alpha curves. For each specification I record the annualized maximum-Sharpe gain $\Delta SR$, rank-area alpha, $R^2_Y$, $IAE$, and the change from the axis-specific CAPM baseline. The exercise separates frontier expansion from pricing a predetermined characteristic coordinate.
	
	\begin{table}[H]
		\centering
		\singlespacing
		\caption{Selected CAPM-plus-one-factor results across characteristic axes}
		\label{tab:cross-axis-factor-scan}
		\footnotesize
		\setlength{\tabcolsep}{3.6pt}
		\begin{tabular}{llrrrrrr}
			\toprule
			Axis
			& Factor
			& $\Delta SR$
			& Rank alpha
			& $t$
			& $R^2_Y$
			& $IAE$
			& $\Delta IAE$ \\
			\midrule
			Value
			& FF HML
			& 0.164 & -26.1 & -2.36 & 0.715 & 26.7 & +2.6 \\
			&
			GFD BE/ME
			& 0.020 & 8.6 & 0.74 & 0.667 & 11.3 & -12.8 \\
			&
			q5 q\_EG
			& 1.240 & 63.9 & 3.03 & 0.057 & 61.8 & +37.7 \\
			\addlinespace
			
			Operating profitability
			& FF RMW
			& 0.263 & -9.4 & -1.01 & 0.507 & 9.8 & -22.0 \\
			&
			GFD OPE/BE
			& 0.178 & 4.9 & 0.41 & 0.349 & 5.2 & -26.7 \\
			&
			q5 q\_ROE
			& 0.496 & -12.3 & -0.86 & 0.261 & 12.8 & -19.0 \\
			&
			q5 q\_EG
			& 1.240 & -22.7 & -1.35 & 0.108 & 25.3 & -6.6 \\
			\addlinespace
			
			Investment
			& FF CMA
			& 0.345 & -5.8 & -0.71 & 0.745 & 9.1 & -40.6 \\
			&
			q5 q\_IA
			& 0.448 & -12.5 & -1.39 & 0.680 & 14.4 & -35.4 \\
			&
			GFD INV\_GR1
			& 0.349 & 14.9 & 1.05 & 0.384 & 15.2 & -34.6 \\
			&
			GFD NOA\_GR1A
			& 0.346 & 12.0 & 0.88 & 0.426 & 12.5 & -37.2 \\
			&
			q5 q\_EG
			& 1.240 & 41.9 & 2.65 & 0.158 & 39.8 & -9.9 \\
			\addlinespace
			
			Momentum
			& FF UMD
			& 0.272 & -13.6 & -1.08 & 0.775 & 18.6 & -53.5 \\
			&
			GFD ret 12--7
			& 0.212 & 11.0 & 0.66 & 0.533 & 11.1 & -61.0 \\
			&
			q5 q\_ROE
			& 0.496 & 30.7 & 1.09 & 0.084 & 30.3 & -41.8 \\
			&
			q5 q\_IA
			& 0.448 & 86.8 & 3.29 & 0.012 & 87.2 & +15.1 \\
			&
			q5 q\_EG
			& 1.240 & 15.2 & 0.58 & 0.042 & 20.6 & -51.5 \\
			\bottomrule
		\end{tabular}
		
		\vspace{0.4em}
		\begin{minipage}{0.96\textwidth}
			\footnotesize
			\emph{Note:} The library comprises Fama--French/Carhart, q-factor-family, and Global Factor Data factors; the reconstructions in \Cref{subsec:prime-factors} are excluded. The table reports standard counterpart factors, related alternatives, and q\_EG as a common high-$\Delta SR$ benchmark. $\Delta SR$ is relative to CAPM. Rank alpha, $IAE$, and $\Delta IAE$ are annualized basis points; $\Delta IAE$ uses unrounded values. CAPM $IAE$ equals 24.1 bp on value, 31.9 on operating profitability, 49.7 on investment, and 72.1 on momentum. Tests use Newey--West standard errors with six monthly lags.
		\end{minipage}
	\end{table}
	
	\begin{figure}[!htbp]
		\centering
		\begin{minipage}{0.48\textwidth}
			\centering
			\includegraphics[width=\linewidth]{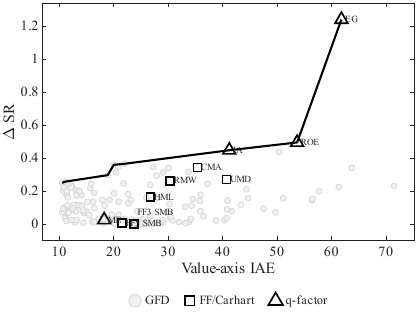}
			\caption*{A. Value}
		\end{minipage}
		\hfill
		\begin{minipage}{0.48\textwidth}
			\centering
			\includegraphics[width=\linewidth]{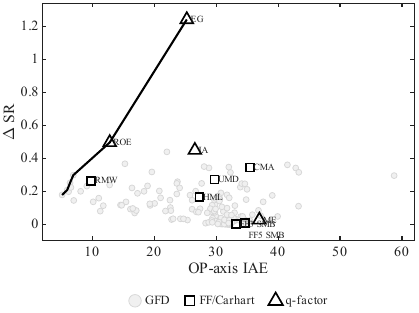}
			\caption*{B. Operating profitability}
		\end{minipage}
		
		\vspace{0.5em}
		
		\begin{minipage}{0.48\textwidth}
			\centering
			\includegraphics[width=\linewidth]{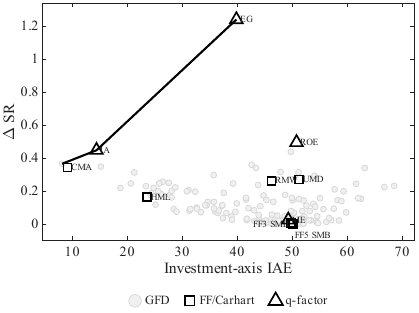}
			\caption*{C. Investment}
		\end{minipage}
		\hfill
		\begin{minipage}{0.48\textwidth}
			\centering
			\includegraphics[width=\linewidth]{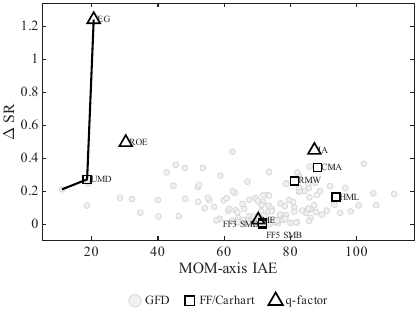}
			\caption*{D. Momentum}
		\end{minipage}
		
		\caption{Axis distortion and maximum-Sharpe gains across candidate factors}
		\label{fig:cross-axis-factor-scan}
		\begin{minipage}{0.94\linewidth}
			\footnotesize
			\emph{Note:} Each panel plots axis-specific $IAE$ against the annualized maximum-Sharpe gain $\Delta SR$ relative to CAPM for 155 factors. Circles denote Global Factor Data factors, squares Fama--French/Carhart factors, and triangles q-factor-family factors. Lines trace the upper-left frontier toward lower $IAE$ and higher $\Delta SR$.
		\end{minipage}
	\end{figure}
	
	Related factors move the same coordinate differently. On value, HML raises $R^2_Y$ to 0.715 but carries the rank alpha from $+24.3$ to $-26.1$ bp and increases $IAE$ to 26.7 bp. GFD BE/ME instead leaves a positive alpha of $+8.6$ bp and lowers $IAE$ to 11.3 bp despite a much smaller $\Delta SR$. Profitability and momentum provide corresponding flattening cases: RMW and OPE/BE reduce OP-axis $IAE$ to 9.8 and 5.2 bp, while UMD and ret 12--7 reduce momentum $IAE$ to 18.6 and 11.1 bp.
	
	Investment displays construction-dependent sign reversal. CMA and q\_IA carry the rank alpha below zero, to $-5.8$ and $-12.5$ bp, whereas INV\_GR1 and NOA\_GR1A leave positive alphas of $+14.9$ and $+12.0$ bp with similarly low $IAE$.\footnote{CAPM augmented with q\_IA alone reproduces the FF5 investment-axis point estimates in \Cref{tab:inv-axis-main} to two decimal places ($-12.53$ bp and $14.36$ bp against $-12.53$ bp and $14.37$ bp). The coincidence is numerical rather than structural: the specifications differ in their remaining factors and in precision, with $R^2_Y$ of 0.680 against 0.761.} None of these rank alphas is significant. Investment information therefore reduces distortion, but whether the curve crosses zero depends on its traded implementation.
	
	The scan also locates q5's momentum fit. q\_ME leaves momentum $IAE$ near the CAPM baseline at 70.2 bp, and q\_IA raises it to 87.2 bp. The reduction comes from q\_ROE and q\_EG, which lower $IAE$ to 30.3 and 20.6 bp. Thus q5's near-zero momentum curve is generated by profitability and expected-growth content rather than a factor formed directly on the momentum order.
	
	\FloatBarrier
	
	\subsection{Sharpe Gains and Axis Destinations}
	\label{subsec:dsr_vs_iae}
	
	Maximum-Sharpe contribution and axis destination are connected but not equivalent. By \eqref{eq:curve_shift} and \eqref{eq:sr2_decomposition}, the absolute appraisal ratio determines frontier expansion and scales the potential correction, while residual alignment and the starting curve determine its sign, shape, and destination.
	
	For the fixed full-sample CAPM benchmark, $\Delta SR$ is strictly increasing in $\lvert IR_g\rvert$. Across 155 candidate factors, its Spearman correlation with $IAE$ is $+0.045$ on value, $+0.034$ on momentum, $-0.362$ on profitability, and $-0.344$ on investment; the corresponding Pearson correlations are $+0.313$, $-0.215$, $-0.213$, and $-0.318$. Frontier gains are therefore nearly rank-unrelated to distortion on value and momentum and only moderately associated with lower distortion on profitability and investment. The divergence between Pearson and Spearman correlations, particularly on momentum, provides no support for a stable linear or monotone mapping between the two criteria.
	
	The starting curve helps explain these differences. CAPM begins with a relatively small value-axis $IAE$ of 24.1 bp, so a large correction can cross zero and increase absolute distortion, whereas investment and momentum begin at 49.7 and 72.1 bp and leave more room for improvement. Even within an axis, however, frontier contribution does not identify the flattest destination. q\_EG produces the largest $\Delta SR$ in the library, 1.240, yet leaves $IAE$ of 61.8 bp on value and 39.8 bp on investment. Factors with much smaller gains---BE/ME, OPE/BE, NOA\_GR1A, and ret 12--7---leave substantially flatter curves on their respective coordinates. Maximum-Sharpe gains therefore measure the potential scale of correction, not where it lands. These comparisons are full-sample results; the value subperiod evidence further shows that the destination associated with a given model need not remain constant through time.
	
	\FloatBarrier
	
	\subsection{Counterpart-Factor Reconstruction}
	\label{subsec:construction_sensitivity}
	
	The appraisal ratio and residual alignment in \eqref{eq:alpha_shift} belong to a traded portfolio relative to the remaining model, not to the characteristic in the abstract. I therefore replace each canonical counterpart factor with the univariate reconstruction defined in \Cref{subsec:prime-factors}: $\mathrm{HML}'$, $\mathrm{RMW}'$, and $\mathrm{CMA}'$ in FF5, and $\mathrm{UMD}'$ in FF6. Letting $F\setminus g$ denote the model without its counterpart factor and $A=\alpha_{F\setminus g}(Y^x)$, \Cref{tab:prime-replacement} compares the canonical and reconstructed factors against the same full-sample base.
	
	\begin{table}[!t]
		\centering
		\singlespacing
		\caption{Replacing canonical counterpart factors with univariate reconstructions}
		\label{tab:prime-replacement}
		\footnotesize
		\setlength{\tabcolsep}{3.2pt}
		\begin{tabular}{llrrrrrrrrr}
			\toprule
			& & \multicolumn{5}{c}{Decomposition of the signed change}
			& \multicolumn{4}{c}{Resulting specification} \\
			\cmidrule(lr){3-7}\cmidrule(lr){8-11}
			Axis
			& $g$
			& $\alpha_{F\setminus g}(g)$
			& $\sigma(\varepsilon_g)$
			& $IR_g$
			& $\operatorname{corr}$
			& $\Delta SA$
			& Rank $\alpha$
			& $p$
			& $R^2_Y$
			& $IAE$ \\
			\midrule
			Value ($A=-40.6$)
			& HML             & $-114.7$ & 7.47\%  & $-0.153$ & 0.739 & $+12.0$ & $-28.6$ & 0.0030 & 0.760 & 28.8 \\
			& $\mathrm{HML}'$ & $-154.6$ & 8.44\%  & $-0.183$ & 0.799 & $+15.5$ & $-25.1$ & 0.0039 & 0.808 & 25.7 \\
			\addlinespace
			OP ($A=+37.8$)
			& RMW             & 479.1 & 7.14\% & 0.671 & 0.661 & $-42.8$ & $-5.0$  & 0.5951 & 0.529 & 6.1 \\
			& $\mathrm{RMW}'$ & 551.9 & 7.66\% & 0.720 & 0.727 & $-50.5$ & $-12.8$ & 0.1707 & 0.606 & 12.7 \\
			\addlinespace
			Investment ($A=+22.8$)
			& CMA             & 298.4 & 4.89\% & 0.610 & 0.730 & $-35.3$ & $-12.5$ & 0.1185 & 0.761 & 14.4 \\
			& $\mathrm{CMA}'$ & 165.4 & 6.55\% & 0.253 & 0.804 & $-16.1$ & $+6.7$  & 0.3634 & 0.820 & 8.5 \\
			\addlinespace
			Momentum ($A=+91.2$)
			& UMD             & 813.7 & 13.99\% & 0.582 & 0.885 & $-84.2$ & $+7.0$  & 0.5536 & 0.805 & 16.1 \\
			& $\mathrm{UMD}'$ & 796.3 & 16.10\% & 0.495 & 0.868 & $-70.2$ & $+21.0$ & 0.1043 & 0.778 & 21.5 \\
			\bottomrule
		\end{tabular}
		
		\vspace{0.4em}
		\begin{minipage}{0.97\textwidth}
			\footnotesize
			\emph{Note:} $F$ is FF5 for value, operating profitability, and investment and FF6 for momentum. Each pair uses the common base $F\setminus g$, with $A=\alpha_{F\setminus g}(Y^x)$ and signed shift $\Delta SA$ from \eqref{eq:alpha_shift}; $A+\Delta SA$ reproduces Rank $\alpha$ up to rounding. Factor alphas, $A$, $\Delta SA$, Rank $\alpha$, and $IAE$ are annualized basis points, and residual volatility is annualized. Rank-alpha p-values use Newey--West standard errors with six monthly lags. Among the reconstructions, $\mathrm{HML}'$ is rejected on all nonlinear functionals, $\mathrm{UMD}'$ on $SUP$ alone, and $\mathrm{RMW}'$ and $\mathrm{CMA}'$ on none.
		\end{minipage}
	\end{table}
	
	\begin{figure}[!t]
		\centering
		\begin{minipage}{0.48\textwidth}
			\centering
			\includegraphics[width=\linewidth]{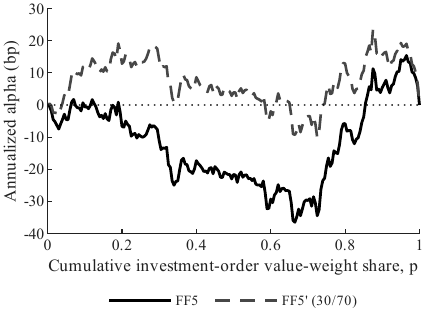}
			\caption*{A. Investment axis: $\mathrm{CMA}\rightarrow\mathrm{CMA}'$}
		\end{minipage}
		\hfill
		\begin{minipage}{0.48\textwidth}
			\centering
			\includegraphics[width=\linewidth]{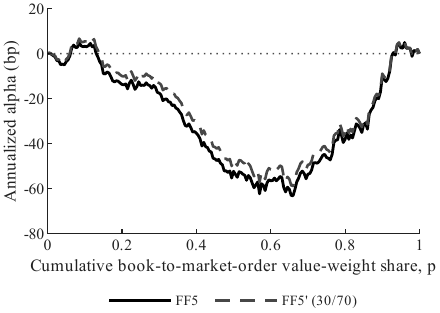}
			\caption*{B. Value axis: $\mathrm{HML}\rightarrow\mathrm{HML}'$}
		\end{minipage}
		
		\caption{Bridge-alpha curves under counterpart-factor reconstruction}
		\label{fig:prime-curves}
		
		\vspace{0.3em}
		\begin{minipage}{0.94\linewidth}
			\footnotesize
			\emph{Note:} Full-sample monthly curves for FF5 and FF5 with the counterpart factor replaced by its univariate reconstruction. Replacement reverses and attenuates the investment curve but leaves value significantly negative. Alphas are annualized basis points.
		\end{minipage}
	\end{figure}
	
	The four axes respond differently. Replacing CMA with $\mathrm{CMA}'$ moves the investment rank alpha from $-12.5$ to $+6.7$ bp, lowers $IAE$ from 14.4 to 8.5 bp, and raises $R^2_Y$ from 0.761 to 0.820. Value remains significantly negative at $-25.1$ bp and is rejected on all nonlinear functionals, while profitability and momentum move away from zero, from $-5.0$ to $-12.8$ bp and from $+7.0$ to $+21.0$ bp. Reconstruction therefore flattens investment, leaves the value rejection intact, and worsens the other two axes rather than acting as a generic flattening device.
	
	The decomposition explains the contrast. On investment, reconstruction lowers the factor alpha against the remaining FF5 factors and raises residual volatility, reducing $IR_g$ from 0.610 to 0.253 despite stronger alignment, from 0.730 to 0.804. The weaker appraisal ratio contracts $\Delta SA$ from $-35.3$ to $-16.1$ bp and moves the common $+22.8$ bp baseline back into positive territory. Under $\mathrm{RMW}'$, both $IR_g$ and alignment rise, making the profitability correction more negative; under $\mathrm{UMD}'$, both weaken, reducing the momentum correction and leaving a larger positive curve.
	
	Value shows why counterpart-factor content does not imply counterpart-factor attribution. HML and $\mathrm{HML}'$ have positive residual alignment but negative appraisal ratios against the remaining FF5 factors, so both shift the common $-40.6$ bp baseline upward, to $-28.6$ and $-25.1$ bp. Removing HML therefore deepens the value distortion: HML mitigates rather than generates it. The subperiod evidence in \Cref{subsec:subperiod} shows that the destination varies through time without changing this narrower full-sample attribution.
	
	These results reinforce the model-by-axis fingerprint. Neither high $R^2_Y$ nor a large maximum-Sharpe gain guarantees a zero curve: the appraisal ratio scales the potential correction, residual alignment governs transmission, and the starting curve determines the destination. Because these objects belong to the factor portfolio and base model, the same signal can land differently when repackaged, and the full-sample replacement effects need not remain constant across subperiods.
	
	\FloatBarrier
	
	\section{Robustness Checks}
	\label{sec:robust}
	
	This section benchmarks the axis results against conventional decile tests and examines robustness to factor implementation, block-bootstrap inference, finite-sample behavior, and grid-covariance singularity.
	
	\subsection{Decile and axis tests}
	\label{subsec:robust_grs_comparison}
	
	For each axis, I estimate ten value-weighted decile alphas, their GRS joint test, and the extreme high-minus-low alpha using the same monthly samples and models as in the axis analysis. The comparison covers the market-capitalization axis of \citet{Shin26} and the four characteristic axes studied here. GRS tests arbitrary joint alpha configurations on supplied portfolios and can pool anomaly families, whereas the axis diagnostic uses the predetermined order through its prefix path, rank-area payoff, and nonlinear functionals. It therefore trades broader omnibus scope for deeper resolution within one coordinate.
	
	\begin{table}[!htbp]
		\centering
		\singlespacing
		\caption{Illustrative disagreements between conventional decile and axis tests}
		\label{tab:grs-axis-comparison}
		\footnotesize
		\setlength{\tabcolsep}{6pt}
		\begin{tabular}{llrrrr}
			\toprule
			Axis
			& Model
			& Decile GRS $p$
			& H--L $p$
			& Rank-area $p$
			& Axis $IAE$ $p$ \\
			\midrule
			Market capitalization
			& FF5
			& 0.8218
			& 0.9831
			& 0.0041
			& 0.0039 \\
			
			Market capitalization
			& q5
			& 0.0107
			& 0.0472
			& 0.6897
			& 0.8686 \\
			
			Value
			& FF5
			& 0.1024
			& 0.1695
			& 0.0030
			& 0.0032 \\
			
			Investment
			& FF5
			& 0.0186
			& 0.7928
			& 0.1185
			& 0.1027 \\
			
			Momentum
			& Carhart
			& 0.0002
			& 0.0007
			& 0.5374
			& 0.2371 \\
			\bottomrule
		\end{tabular}
		
		\vspace{0.4em}
		\begin{minipage}{0.95\textwidth}
			\footnotesize
			\emph{Note:} GRS uses ten value-weighted deciles formed along each order, and H--L is the extreme-decile spread. Rank-area p-values use Newey--West standard errors with six monthly lags; axis $IAE$ p-values use the model-specific HAC-GP null in \Cref{subsec:finite_axis_implementation}. Because the procedures differ in test assets, alternatives, and calibration, the comparison is diagnostic rather than a ranking of unconditional power. The full 30-cell comparison appears in \Cref{tab:app-grs-axis-full}.
		\end{minipage}
	\end{table}
	
	Across the 30 model--axis cells, GRS and axis $IAE$ agree in 20, with ten joint rejections and ten joint non-rejections. The remaining disagreements divide evenly between five GRS-only and five axis-only rejections, so neither procedure is uniformly more aggressive. The capitalization axis displays both directions: FF3, FF5, and FF6 pass GRS but fail rank area and $IAE$, whereas q5 fails GRS and the extreme spread but passes both axis statistics.
	
	The other disagreements reflect the alternatives each procedure targets. FF5 and FF6 on value leave broad negative distortions through the interior of the order that decile tests miss. Carhart and FF6 on momentum, and FF5 on investment, instead contain portfolio-level errors detected by GRS that do not accumulate into large rank-area or $IAE$ statistics. GRS is sensitive to sparse, oscillating, endpoint-specific, and other finite-dimensional configurations; the axis statistics concentrate power on errors that persist cumulatively along the coordinate.
	
	The procedures nevertheless identify the same clear failures and passes in most cells, while GRS and the extreme spread agree in 28 of 30. Factor reconstruction provides a further common check: replacing CMA with $\mathrm{CMA}'$ raises the FF5 investment-decile GRS p-value from 0.0186 to 0.291, lowers axis $IAE$ from 14.4 to 8.5 bp, and moves the rank-area alpha from $-12.5$ to $+6.7$ bp. Alternative packaging can therefore improve both conventional and order-aware pricing, even though the corresponding HML reconstruction leaves value significantly negative.
	
	\subsection{Anomaly-factor replication checks}
	\label{subsec:robust_factor_replication}
	
	The accounting axes and factor replications use the same timing convention: at June $Y$ formation, I select fiscal year $Y-1$, with investment growth additionally using total assets from $Y-2$. Using the same CRSP--Compustat merge, return accounting, annual formation schedule, and published breakpoints as in the main analysis, I reconstruct the standard $2\times3$ factors:
	
		{
		\setlength{\jot}{2pt}
		\[
		\begin{aligned}
			HML_t
			&=
			\tfrac{1}{2}(SH_t+BH_t)
			-
			\tfrac{1}{2}(SL_t+BL_t),
			\\
			RMW_t
			&=
			\tfrac{1}{2}(SR_t+BR_t)
			-
			\tfrac{1}{2}(SW_t+BW_t),
			\\
			CMA_t
			&=
			\tfrac{1}{2}(SLoINV_t+BLoINV_t)
			-
			\tfrac{1}{2}(SHiINV_t+BHiINV_t),
			\\
			INV_{i,Y}
			&=
			\frac{AT_{i,Y-1}-AT_{i,Y-2}}{AT_{i,Y-2}}.
		\end{aligned}
		\]
	}
	
	Exact equality with the published series is not expected because of historical data revisions and differences in book equity, delisting treatment, and reinvestment conventions. The reconstructed factors serve only as implementation diagnostics; all reported bridge regressions use published factors. Component correlations directly assess stock assignment and portfolio formation, while factor correlations, betas, and intercepts assess the resulting long--short combinations.
	
	\begin{table}[H]
		\centering
		\caption{Anomaly-factor replication}
		\label{tab:robust-replication}
		\footnotesize
		\begin{tabular}{lcccccccc}
			\toprule
			& \multicolumn{3}{c}{Monthly}
			& \multicolumn{3}{c}{Daily}
			& \multicolumn{2}{c}{Component portfolios} \\
			\cmidrule(lr){2-4}
			\cmidrule(lr){5-7}
			\cmidrule(lr){8-9}
			Factor
			& Corr. & Beta & $\alpha$ (bp)
			& Corr. & Beta & $\alpha$ (bp)
			& Monthly & Daily \\
			\midrule
			HML
			& 0.970 & 1.018 & $-45.7$
			& 0.961 & 0.978 & $-33.9$
			& 0.979--0.998 & 0.979--0.997
			\\
			RMW
			& 0.955 & 0.940 & $-5.0$
			& 0.936 & 0.932 & $-8.7$
			& 0.988--0.997 & 0.988--0.997
			\\
			CMA
			& 0.943 & 0.923 & $21.4$
			& 0.908 & 0.913 & $19.3$
			& 0.985--0.997 & 0.980--0.996
			\\
			\bottomrule
		\end{tabular}
		
		\vspace{0.4em}
		\begin{minipage}{0.95\linewidth}
			\footnotesize
			\emph{Note:} Corr., Beta, and $\alpha$ compare each reconstructed factor with its published counterpart; $\alpha$ is an annualized intercept in basis points. The last two columns report correlations across the six underlying $2\times3$ portfolios. The reconstructed factors are used only in this table. Monthly and daily samples contain 696 and 14,598 observations, respectively.
		\end{minipage}
	\end{table}
	
	Factor correlations are 0.943--0.970 monthly and 0.908--0.961 daily, with betas of 0.913--1.018, while component correlations are uniformly higher at 0.979--0.998 monthly and 0.979--0.997 daily. The $-45.7$ bp HML intercept shows that small component discrepancies need not cancel in the final spread, but cannot affect the value-axis results because those regressions use published HML. The close component agreement weakens explanations based on merging, fiscal-year timing, breakpoint assignment, or portfolio formation. UMD requires no corresponding accounting check because it uses neither Compustat data nor fiscal-year matching.
	
	\FloatBarrier
	
	\subsection{Block-bootstrap critical values}
	\label{subsec:robust_block_bootstrap}
	
	As a distributional check on HAC-GP inference, I recompute nonlinear-functional critical values using a circular six-month moving-block residual bootstrap. For each model--axis pair, write $\bm D_t=\bm\alpha+f_t'\bm B+\bm\varepsilon_t$. I center the $L$-dimensional residual vector grid point by grid point and resample it jointly, preserving contemporaneous cross-grid and short-range serial dependence. The zero-curve null is imposed through
	\begin{equation}
		\bm D^{*(b)}_t
		=
		f_t'\widehat{\bm B}
		+
		\widetilde{\bm\varepsilon}^{*(b)}_t,
		\qquad
		b=1,\ldots,B=4{,}999,
		\label{eq:block_bootstrap_null_dgp}
	\end{equation}
	and each resample re-estimates the full regression grid. The upper-tail p-value is
	\begin{equation}
		\widehat p^{\,BB}_F
		=
		\frac{
			1+
			\sum_{b=1}^{B}
			\mathbf 1
			\left\{
			F(\widehat{\bm\alpha}^{*(b)})
			\geq
			F(\widehat{\bm\alpha})
			\right\}
		}{
			B+1
		},
		\qquad
		F\in\{IAE,ISE,SUP\}.
		\label{eq:block_bootstrap_pvalue}
	\end{equation}
	Unlike HAC-GP, the bootstrap bypasses the estimated long-run covariance matrix, its positive-semidefinite adjustment, and Gaussian simulation.
	
	\begin{table}[!htbp]
		\centering
		\singlespacing
		\caption{HAC-GP and block-bootstrap rejection maps}
		\label{tab:robust-block-bootstrap-summary}
		\footnotesize
		\begin{tabular}{lcccc}
			\toprule
			& & \multicolumn{2}{c}{Cells rejected at 5\%} & \\
			\cmidrule(lr){3-4}
			Axis & Cells & HAC-GP & Bootstrap & Agreement \\
			\midrule
			Value                   & 18 & 12 & 12 & 18/18 \\
			Operating profitability & 18 &  3 &  4 & 15/18 \\
			Investment              & 18 &  5 &  9 & 14/18 \\
			Momentum                & 18 &  9 &  9 & 18/18 \\
			\midrule
			Total                   & 72 & 29 & 34 & 65/72 \\
			\bottomrule
		\end{tabular}
		
		\vspace*{0.4em}
		\begin{minipage}{0.95\linewidth}
			\footnotesize
			\emph{Note:} Each axis contains six models and three nonlinear functionals. HAC-GP results are from \Cref{tab:value-axis-main}--\Cref{tab:mom-axis-main}; bootstrap results use 4,999 circular six-month residual resamples under \eqref{eq:block_bootstrap_null_dgp}.
		\end{minipage}
	\end{table}
	
	The procedures agree in 65 of 72 cells and in every value and momentum cell. The seven boundary disagreements are concentrated in profitability and investment, where the bootstrap produces somewhat more rejections. It therefore strengthens selected negative-investment results without changing the main value or momentum conclusions, while qualifying any blanket claim that q5 is never rejected.
	
	\subsection{Finite-sample calibration}
	\label{subsec:robust_size_power}
	
	Holding the empirical monthly factors, fitted loadings, and centered residual curves fixed, I simulate
	\begin{equation}
		\bm\alpha^x_m(\lambda)
		=
		\lambda\widehat{\bm\alpha}^x_m,
		\qquad
		\lambda\in\{0,0.25,0.5,0.75,1\},
		\label{eq:power_empirical_shape_alternative}
	\end{equation}
	within the bootstrap environment of \eqref{eq:block_bootstrap_null_dgp}. The case $\lambda=0$ measures size, while positive values preserve each estimated curve's shape and location but scale its magnitude. For each $\lambda$, I generate 4,999 evaluation samples. HAC-GP critical values use an independent set of 50,000 Gaussian draws, while bootstrap critical values use an independent set of 4,999 null resamples. Because the alternatives remain model-specific, rejection frequencies are not common-effect-size comparisons across models or axes.
	
	\begin{table}[!htbp]
		\centering
		\singlespacing
		\caption{Finite-sample size and power calibration over alternative strength}
		\label{tab:robust-size-power-calibration}
		\footnotesize
		\setlength{\tabcolsep}{4.5pt}
		\resizebox{\textwidth}{!}{
			\begin{tabular}{llrrrrrl}
				\toprule
				Axis & Critical value
				& $\lambda=0$
				& $\lambda=0.25$
				& $\lambda=0.50$
				& $\lambda=0.75$
				& $\lambda=1$
				& $IAE$ at $\lambda=1$ \\
				\midrule
				Value & HAC-GP
				& 5.3\% & 10.3\% & 25.6\% & 47.7\% & 66.0\% & 22--29 bp \\
				& Block bootstrap
				& 5.0\% & 9.8\% & 24.5\% & 46.5\% & 65.5\% & \\
				\addlinespace
				
				Operating profitability & HAC-GP
				& 4.5\% & 6.0\% & 11.0\% & 20.2\% & 32.9\% & 6--32 bp \\
				& Block bootstrap
				& 5.0\% & 6.6\% & 12.0\% & 21.6\% & 34.6\% & \\
				\addlinespace
				
				Investment & HAC-GP
				& 3.8\% & 6.3\% & 14.9\% & 31.1\% & 51.7\% & 14--50 bp \\
				& Block bootstrap
				& 5.2\% & 8.3\% & 18.9\% & 36.8\% & 57.0\% & \\
				\addlinespace
				
				Momentum & HAC-GP
				& 3.5\% & 7.9\% & 22.3\% & 42.1\% & 56.5\% & 15--93 bp \\
				& Block bootstrap
				& 4.7\% & 10.2\% & 27.1\% & 46.6\% & 60.3\% & \\
				\bottomrule
			\end{tabular}
		}
		
		\vspace*{0.4em}
		\begin{minipage}{0.95\linewidth}
			\footnotesize
			\emph{Note:} Entries average 5\% rejection frequencies across six models and three nonlinear functionals. Positive $\lambda$ scales each model's empirical curve, so axis averages do not represent common alternatives. HAC-GP uses 50,000 independent Gaussian draws; bootstrap quantities use separate sets of 4,999 circular six-month resamples.
		\end{minipage}
	\end{table}
	
	Null rejection remains near 5\% under both methods, and power rises monotonically with $\lambda$. The lower profitability average partly reflects the nearly flat FF5 and FF6 planted alternatives rather than test failure.
	
	The FF6--q5 momentum comparison illustrates model-specific resolution. Their observed $IAE$ values are similar, at 16.1 and 17.8 bp, but HAC-GP rejection at $\lambda=1$ is 18\% for FF6 and 4\% for q5. FF6 has $R^2_Y=0.805$ and a 23.4 bp critical value, versus 0.134 and 52.4 bp for q5. Although differences in planted curve shape prevent attributing the power gap solely to covariance fit, the comparison shows that similarly sized distortions can have sharply different evidentiary content. The calibration establishes local power against empirical-shape alternatives, not uniform power against all possible curves.
	
	\subsection{Rank deficiency of the grid covariance matrix}
	\label{subsec:robust_singularity}
	
	The $201\times201$ grid covariance matrix $\widehat{\Omega}^x_m$ is structurally singular. Endpoint bridges are identically zero, and adjacent interior bridges differ only through stocks near the moving cutoff, inducing substantial cross-grid dependence.
	
	The inference procedure accommodates this directly. $IAE$, $ISE$, and $SUP$ require only a positive-semidefinite square root of $\widehat{\Omega}_{+}/T$, never its inverse; the only inversion is the small factor second-moment matrix $Q_m$ in \eqref{eq:axis_alpha_influence}. Moreover,
	\[
	\frac{
		\left\|
		\widehat{\Omega}_{+}
		-
		\widehat{\Omega}
		\right\|_{F}
	}{
		\left\|
		\widehat{\Omega}
		\right\|_{F}
	}
	=
	1.4\times10^{-10},
	\]
	so the positive-semidefinite adjustment is numerically negligible. The bootstrap bypasses $\widehat{\Omega}$ entirely yet agrees with HAC-GP in 65 of 72 cells. Rank deficiency is therefore an expected feature of the bridge path, not evidence that numerical instability drives the results.
	
	\section{Discussion}
	\label{sec:discuss}
	
	\paragraph{Interpreting where corrections land.}
	\label{subsec:discuss_overcorrection}
	
	Counterpart factors generally shift their characteristic-axis curves downward, but they do not land at the same destination. RMW and UMD bring profitability and momentum close to zero, HML-based models leave value significantly negative, and investment-factor models cross zero with weaker inferential support. By \eqref{eq:curve_shift}, the appraisal ratio scales the potential movement, residual alignment governs its transmission along the order, and the starting curve determines whether the enlarged model remains undercorrected, becomes flat, or overshoots. A downward shift is therefore not sufficient evidence of successful pricing, nor can the destination be attributed mechanically to the factor sharing the characteristic's name. Removing HML from FF5 makes the value curve more negative, whereas CMA supplies enough correction to carry investment below zero. The axis diagnostic locates the remaining error; attribution requires addition, removal, or replacement experiments.
	
	\paragraph{Time variation in the value destination.}
	\label{subsec:discuss_time_variation}
	
	The value overshoot is recurrent but not time-invariant. Rolling estimates move from an early near-zero-to-positive interval to a deep late-1990s trough, subsequent attenuation, and renewed negative distortion near the sample end, without identifying a unique structural break. The comparison between FF5 and q5 also separates estimated direction from test resolution: both frequently produce negative rank-area estimates, but q5 rejects much less often because its standard errors and functional critical values are wider. Smoothly rising capitalization coverage does not reproduce the nonmonotone chronology, although lower early coverage remains a qualification on the scope of the earliest estimates.
	
	\paragraph{Functional geometry and test resolution.}
	\label{subsec:discuss_functional}
	
	The four summaries emphasize different features of the same curve: $SA$ records direction, $IAE$ absolute magnitude, $ISE$ concentration, and $SUP$ the largest local deviation. Their occasional disagreements therefore reflect curve geometry rather than contradictory evidence. Inference adds a second dimension because each nonlinear functional is calibrated to the model's own cross-grid residual covariance. A noisy curve faces wider null bands, so non-rejection combines estimated magnitude with available resolution rather than constituting an unconditional pricing pass.
	
	\paragraph{Relation to spanning and existing tests.}
	\label{subsec:discuss_positioning}
	
	The characteristic-axis diagnostic complements rather than replaces spanning and conventional specification tests. Maximum-Sharpe spanning asks whether an added factor expands the attainable frontier, GRS tests whether a model prices a supplied finite asset set, and the axis diagnostic asks whether zero alpha holds on the return space generated by one predetermined order. Their empirical disagreements run in both directions: broad ordered distortions can be detected by the axis procedure despite passing conventional decile tests, while localized or nonmonotone decile errors may trigger GRS without accumulating along the coordinate. GRS retains broader omnibus scope over arbitrary asset systems and cross-anomaly dependence; the axis procedure gives up that role for deeper resolution within one prespecified coordinate. Spanning remains distinct from both. The factor scan shows that large frontier gains need not flatten the relevant axis because appraisal ratios determine only the potential scale of correction, while signed residual alignment and the starting curve determine its destination.
	
	\paragraph{What a rejection identifies, and what it does not.}
	\label{subsec:discuss_scope}
	
	An axis rejection concerns a particular traded model on a particular coordinate-generated return space, not the characteristic in the abstract. Canonical factors may be size-balanced or formed through multidimensional sorts, whereas the empirical axes are cap-weighted; this construction mismatch may explain a rejection but does not remove the resulting pricing error. Reconstructed counterpart factors materially move all four curves, confirming that the same signal can reach a different destination when packaged as a different traded portfolio, although the replacement experiments do not isolate any single design feature. The conclusions remain local in three respects. Accounting-axis results concern valid CRSP--Compustat subuniverses, and gate-failing models are interpreted only on the internal zero-investment bridge space. Results are model-by-axis specific and provide neither a global ranking nor a cross-axis omnibus test. Finally, population pricing concerns only the span generated by the chosen order, while empirical non-rejection means only that no distortion is detected on the implemented grid at the available resolution. This restricted scope is the purpose of the diagnostic: to identify where factor correction lands without claiming validity outside the prespecified coordinate.
	
	\section{Conclusion}
	\label{sec:conclusion}
	
	Maximum-Sharpe spanning measures frontier expansion, not where a factor leaves pricing errors along the characteristic it is meant to price. This paper develops a characteristic-axis integral diagnostic for that complementary question. A predetermined characteristic generates a full path of cap-weighted prefix portfolios, and subtracting equal exposure to the corresponding valid-universe aggregate produces a bridge-alpha curve. Together, the aggregate gate and zero-curve restriction price the prefix, tail, interval, and step-function returns generated by the order without relying on a few selected bins.
	
	Counterpart factors generally shift these curves downward, but they land differently. HML-based models leave value significantly negative, investment factors also cross zero with weaker evidence, and RMW and UMD bring profitability and momentum close to zero. The value destination varies through time rather than reflecting one isolated episode, while q5 shows that similar point estimates can yield different conclusions when residual uncertainty and functional null bands differ.
	
	Neither covariance fit nor maximum-Sharpe contribution determines the destination. The appraisal ratio scales the potential correction, but residual alignment and the starting curve determine whether the enlarged model undercorrects, flattens, or overshoots. The outcome also belongs to the traded implementation: removing or repackaging counterpart factors can materially move the same characteristic axis.
	
	The interpretation remains local. A rejection concerns a particular traded model on a prespecified coordinate-generated return space, not global model failure or invalidity of the underlying characteristic premium. Population pricing extends only to the subspace generated by that order, and empirical non-rejection means only that no distortion is detected at the implemented resolution. Within those limits, the diagnostic reveals where factor corrections land, how precisely those destinations are identified, and how they change through time.
	
	\vspace*{2em}
	
	\paragraph{Funding}
	This research did not receive any specific grant from funding agencies in the public, commercial, or not-for-profit sectors.
	
	\paragraph{Declaration of AI usage} 
	During the preparation of this manuscript, the author used ChatGPT (OpenAI) and Claude (Anthropic) for language refinement and structural clarity. All outputs were reviewed and edited by the author, who takes full responsibility for the content.
	
	\paragraph{Declaration of interest}
	The author declares no competing interests.
	
	\newpage
	\begin{appendices}
		
	\section{Rolling-window summary for the value axis}
	\label{app:value_subperiod_full}
	
	This appendix condenses the rolling-window evidence underlying \Cref{fig:value-rolling-subperiod}. Rather than reproducing every model--window cell, \Cref{tab:value-rolling-summary} summarizes the complete set of 49 overlapping 120-month windows for the models most relevant to the full-sample overshoot and precision comparison. \Cref{tab:value-rolling-selected} then reports selected windows marking the main changes visible in the figure. The complete machine-readable results, including all models, aggregate-gate tests, rank-area standard errors, and functional statistics, are included in the replication archive.
	
	\begin{table}[H]
		\centering
		\singlespacing
		\caption{Summary of rolling ten-year value-axis results}
		\label{tab:value-rolling-summary}
		\small
		\setlength{\tabcolsep}{5.0pt}
		\begin{tabular}{lrrrrrr}
			\toprule
			Model
			& Gate-valid
			& Negative
			& $IAE$ reject
			& Most negative
			& Final-window
			& Final $p(IAE)$ \\
			\midrule
			FF5
			& 44
			& 34 (77.3\%)
			& 16 (36.4\%)
			& $-92.7$
			& $-58.6$
			& 0.002 \\
			
			FF6
			& 43
			& 34 (79.1\%)
			& 11 (25.6\%)
			& $-74.4$
			& $-55.8$
			& 0.004 \\
			
			q5
			& 47
			& 28 (59.6\%)
			& 3 (6.4\%)
			& $-93.5$
			& $-43.0$
			& 0.154 \\
			\bottomrule
		\end{tabular}
		
		\vspace{0.4em}
		\begin{minipage}{0.94\linewidth}
			\footnotesize
			\emph{Note:} The sample contains 49 overlapping 120-month windows from 1967--1976 through 2015--2024. ``Gate-valid'' counts windows in which the characteristic-valid aggregate passes at 5\%. ``Negative'' and ``$IAE$ reject'' are computed only among gate-valid windows. Rank-area alphas are annualized basis points. ``Most negative'' reports the minimum rank-area alpha across gate-valid windows. The final window is 2015--2024. Because adjacent windows overlap by 108 months, the counts describe persistence rather than independent replications.
		\end{minipage}
	\end{table}
	
	The summary separates the prevalence of a negative point estimate from the precision with which it is identified. FF5 and FF6 produce negative destinations in roughly four fifths of their aggregate-valid windows and reject the functional null in approximately one quarter to one third. q5 also produces a negative estimate in a majority of valid windows, including a late-1990s minimum comparable to that of FF5, but rejects in only three windows. Its lower rejection frequency therefore reflects substantially wider sampling uncertainty and functional null bands, not the absence of economically large negative estimates.
	
	\begin{table}[H]
		\centering
		\singlespacing
		\caption{Selected rolling ten-year value-axis windows}
		\label{tab:value-rolling-selected}
		\small
		\setlength{\tabcolsep}{4.5pt}
		\begin{tabular}{lrrrr}
			\toprule
			Window
			& CAPM
			& FF3
			& FF5
			& q5 \\
			\midrule
			1967--1976
			& $+56.5\;[.294]$
			& $-44.7\;[.133]$
			& $-30.3\;[.185]$
			& $-19.2\;[.540]$ \\
			
			1974--1983
			& $\mathbf{+135.0}\;[<.001]$
			& $+20.9\;[.402]$
			& $+30.7\;[.258]$
			& $+40.1\;[.407]$ \\
			
			1980--1989
			& $\mathbf{+89.6}\;[.006]$
			& $-14.0\;[.522]$
			& $-18.5\;[.455]$
			& $-2.0\;[.603]$ \\
			
			1991--2000
			& $+57.9\;[.413]$
			& $\mathbf{-85.6}\;[<.001]$
			& $\mathbf{-92.7}\;[<.001]$
			& $\mathbf{-93.5}\;[.023]$ \\
			
			1999--2008
			& $+95.6\;[.133]$
			& $+8.0\;[.900]$
			& $-5.6\;[.913]$
			& $+15.5\;[.814]$ \\
			
			2003--2012
			& $+0.8\;[.985]$
			& $+3.2\;[.950]$
			& $+4.4\;[.779]$
			& $+12.3\;[.860]$ \\
			
			2009--2018
			& $\mathbf{-79.2}\;[.033]$
			& $-27.8\;[.065]$
			& $\mathbf{-29.1}\;[.027]$
			& $-28.3\;[.311]$ \\
			
			2012--2021
			& $-89.6\;[.126]$
			& $\mathbf{-39.1}\;[.039]$
			& $\mathbf{-33.3}\;[.046]$
			& $-29.1\;[.323]$ \\
			
			2015--2024
			& $-104.2^{\ddagger}\;[.123]$
			& $-67.3^{\ddagger}\;[.001]$
			& $\mathbf{-58.6}\;[.002]$
			& $-43.0\;[.154]$ \\
			\bottomrule
		\end{tabular}
		
		\vspace{0.4em}
		\begin{minipage}{0.94\linewidth}
			\footnotesize
			\emph{Note:} Each cell reports the annualized rank-area alpha in basis points, followed by the HAC--GP $p(IAE)$ in brackets. Bold entries reject the functional null at 5\% while passing the aggregate gate. $^{\ddagger}$~The characteristic-valid aggregate rejects at 5\%, so the corresponding axis statistic is descriptive rather than unconditional evidence. The windows are selected to represent the early sample, the temporary positive destination, the late-1990s negative trough, the subsequent near-zero interval, and the late-sample re-emergence. They are not additional independent tests.
		\end{minipage}
	\end{table}
	
	The HML-based models move from a temporary positive or near-zero destination early in the sample to a late-1990s negative trough, return close to zero, and become negative again near the end. q5 often moves in the same direction but rejects less frequently because of its wider null distribution, supporting a time-varying rather than constant full-sample destination.
	
	\FloatBarrier
	
	\newpage
	
	\section{Full Comparison with Conventional Decile Tests}
	\label{app:grs_axis_full}
	
	\Cref{tab:app-grs-axis-full} reports the complete comparison underlying \Cref{subsec:robust_grs_comparison}. The capitalization panel extends the comparison to the cap-axis benchmark studied in \citet{Shin26}, while the remaining panels cover the four characteristic axes examined in this paper. For each of the five axes and six benchmark models, the table presents the
	conventional GRS p-value for ten value-weighted deciles, the Newey--West p-value for the extreme-decile spread, the Newey--West p-value for the rank-area alpha, and the model-specific HAC-GP p-value for axis $IAE$.
	
	\begin{table}[!htbp]
		\centering
		\singlespacing
		\caption{Full comparison of conventional decile and axis tests}
		\label{tab:app-grs-axis-full}
		\scriptsize
		\setlength{\tabcolsep}{8pt}
		\begin{tabular}{lrrrr}
			\toprule
			Model
			& GRS $p$
			& H--L $p$
			& Rank-area $p$
			& Axis $IAE$ $p$ \\
			\midrule
			
			\multicolumn{5}{l}{\textit{Panel A: Market capitalization}} \\
			CAPM
			& 0.7275
			& 0.8190
			& 0.5552
			& 0.5814 \\
			
			FF3
			& 0.0899
			& 0.1440
			& \textbf{0.0050}
			& \textbf{0.0063} \\
			
			Carhart
			& \textbf{0.0324}
			& 0.1601
			& \textbf{0.0092}
			& \textbf{0.0102} \\
			
			FF5
			& 0.8218
			& 0.9831
			& \textbf{0.0041}
			& \textbf{0.0039} \\
			
			FF6
			& 0.5060
			& 0.9852
			& \textbf{0.0065}
			& \textbf{0.0067} \\
			
			q5
			& \textbf{0.0107}
			& \textbf{0.0472}
			& 0.6897
			& 0.8686 \\
			
			\addlinespace[0.6em]
			\multicolumn{5}{l}{\textit{Panel B: Value}} \\
			CAPM
			& 0.0673
			& 0.2476
			& 0.3226
			& 0.3205 \\
			
			FF3
			& \textbf{0.0053}
			& \textbf{0.0058}
			& \textbf{0.0050}
			& \textbf{0.0041} \\
			
			Carhart
			& \textbf{0.0119}
			& \textbf{0.0133}
			& \textbf{0.0037}
			& \textbf{0.0029} \\
			
			FF5
			& 0.1024
			& 0.1695
			& \textbf{0.0030}
			& \textbf{0.0032} \\
			
			FF6
			& 0.1506
			& 0.2469
			& \textbf{0.0030}
			& \textbf{0.0025} \\
			
			q5
			& 0.1814
			& 0.8220
			& 0.3033
			& 0.1811 \\
			
			\addlinespace[0.6em]
			\multicolumn{5}{l}{\textit{Panel C: Operating profitability}} \\
			CAPM
			& \textbf{0.0087}
			& \textbf{0.0045}
			& 0.0520
			& 0.0554 \\
			
			FF3
			& \textbf{\ensuremath{<}0.0001}
			& \textbf{0.0008}
			& \textbf{0.0495}
			& \textbf{0.0449} \\
			
			Carhart
			& \textbf{0.0004}
			& \textbf{0.0010}
			& 0.0518
			& \textbf{0.0493} \\
			
			FF5
			& 0.0524
			& 0.7814
			& 0.5951
			& 0.8029 \\
			
			FF6
			& 0.1640
			& 0.6898
			& 0.5061
			& 0.6986 \\
			
			q5
			& 0.5959
			& 0.3943
			& 0.2091
			& 0.2049 \\
			
			\addlinespace[0.6em]
			\multicolumn{5}{l}{\textit{Panel D: Investment}} \\
			CAPM
			& \textbf{0.0021}
			& \textbf{0.0015}
			& \textbf{0.0028}
			& \textbf{0.0025} \\
			
			FF3
			& \textbf{0.0399}
			& \textbf{0.0253}
			& \textbf{0.0261}
			& \textbf{0.0239} \\
			
			Carhart
			& 0.1084
			& 0.1354
			& 0.1613
			& 0.1636 \\
			
			FF5
			& \textbf{0.0186}
			& 0.7928
			& 0.1185
			& 0.1027 \\
			
			FF6
			& 0.0752
			& 0.8735
			& 0.0876
			& 0.0861 \\
			
			q5
			& 0.3214
			& 0.7107
			& 0.0976
			& 0.0909 \\
			
			\addlinespace[0.6em]
			\multicolumn{5}{l}{\textit{Panel E: Momentum}} \\
			CAPM
			& \textbf{\ensuremath{<}0.0001}
			& \textbf{\ensuremath{<}0.0001}
			& \textbf{0.0014}
			& \textbf{0.0016} \\
			
			FF3
			& \textbf{\ensuremath{<}0.0001}
			& \textbf{\ensuremath{<}0.0001}
			& \textbf{\ensuremath{<}0.0001}
			& \textbf{\ensuremath{<}0.0001} \\
			
			Carhart
			& \textbf{0.0002}
			& \textbf{0.0007}
			& 0.5374
			& 0.2371 \\
			
			FF5
			& \textbf{\ensuremath{<}0.0001}
			& \textbf{\ensuremath{<}0.0001}
			& \textbf{0.0012}
			& \textbf{0.0017} \\
			
			FF6
			& \textbf{0.0007}
			& \textbf{0.0042}
			& 0.5536
			& 0.1929 \\
			
			q5
			& 0.1276
			& 0.9911
			& 0.9028
			& 0.5315 \\
			\bottomrule
		\end{tabular}
		
		\vspace{0.4em}
		\begin{minipage}{0.95\textwidth}
			\footnotesize
			\emph{Note:} The sample contains 696 monthly observations from January 1967 through December 2024. GRS tests the joint zero-alpha restriction for ten value-weighted deciles formed along the corresponding order. H--L denotes the extreme-decile spread, with its sign oriented to match the direction of the order; its two-sided p-value is invariant to that orientation. H--L and rank-area tests use Newey--West standard errors with six monthly lags. Axis $IAE$ uses the model-specific HAC-GP null described in \Cref{subsec:finite_axis_implementation}. Bold entries reject at the 5\% level. The classical finite-sample GRS $F$ distribution relies on its standard i.i.d.\ multivariate-normal regression-error assumption, whereas H--L, rank-area, and axis-$IAE$ inference allows serial dependence through HAC methods. The table therefore compares each procedure under its conventional implementation rather than holding the inference method fixed.
		\end{minipage}
	\end{table}
	
	Across the 30 model--axis cells, GRS and axis $IAE$ agree in 20 cases: both reject in ten and neither rejects in ten. The remaining ten disagreements are exactly balanced across directions. GRS alone rejects in five cells---q5 on market capitalization, CAPM on operating profitability, FF5 on investment, and Carhart and FF6 on momentum---whereas axis $IAE$ alone rejects in five---FF3, FF5, and FF6 on market capitalization and FF5 and FF6 on value. GRS and the extreme-decile test reach the same 5\% decision in 28 of 30 cells. The two exceptions are Carhart on market capitalization and FF5 on investment, where GRS rejects the joint decile restriction but the extreme spread does not. The expanded map therefore combines broad agreement with exactly symmetric two-way disagreement, reinforcing that neither procedure is uniformly more aggressive than the other.
	
	\end{appendices}
	
	\singlespacing

\end{document}